\documentclass{article}
\usepackage{fullpage}
\usepackage{cite}
\usepackage{hyperref}
\usepackage{verbatim}
\usepackage{tikz}

\usepackage{colortbl}

\usepackage{paralist}
\usepackage{pdfpages}
\usepackage{enumitem}
\usepackage{booktabs}
\usepackage{color,xcolor}
\usepackage{graphicx,color,eso-pic}
\usepackage{amsmath,amssymb,stmaryrd}
\usepackage{amsmath,amsthm,amstext,amssymb,amsfonts,latexsym}

\usepackage{tikz}
\usetikzlibrary{arrows,automata}
\usetikzlibrary{shapes}
\usetikzlibrary{decorations.pathmorphing}
\usetikzlibrary{decorations.pathreplacing}

\usepackage[
	lambda,
	advantage,
	operators,
	sets,
	adversary,
	landau,
	probability,
	notions,
	logic,
	ff,
	mm,
	primitives,
	events,
	complexity,
	asymptotics,
	keys]
{cryptocode}

\newcommand{\pr}[1]{\text{Pr} \left[#1 \right]}

\usepackage{boxedminipage}
\usepackage{url}
\usepackage{graphicx,proof}
\usepackage{graphicx}
\usepackage[bf]{caption}
\usepackage{subcaption}
\usepackage{color}
\usepackage{xspace}
\usepackage{multirow}
\usepackage{amsmath,amstext,amssymb,amsfonts,latexsym}
\usepackage{wrapfig}

{
\newtheorem{theorem}{Theorem}[section]

\newtheorem{corollary}[theorem]{Corollary}
\newtheorem{definition}[theorem]{Definition}
\newtheorem{lemma}[theorem]{Lemma}
\newtheorem{fact}[theorem]{Fact}

\newtheorem{claim*}[theorem]{Claim}

\newtheorem*{remark}{Remark}
}

\newcommand{\X}{\mathcal{X}}

\newcommand{\A}{\mathcal{A}}
\newcommand{\B}{\mathcal{B}}
\newcommand{\C}{\mathcal{C}}
\newcommand{\G}{\mathcal{G}}

\newcommand{\Bit}{\lbrace 0,1 \rbrace}

\renewcommand{\poly}{\mathsf{poly}}

\newcommand{\braces}[1]{\lbrace #1 \rbrace}
\newcommand{\KL}[2]{KL( #1 \; || \; #2)}

\newcommand{\Enc}{\mathsf{Enc}}
\newcommand{\Dec}{\mathsf{Dec}}

\newcommand{\len}{\text{len}}
\newcommand{\Pred}{\mathsf{Pred}}

\newcommand{\Sample}{\mathsf{Sample}}

\newcommand{\End}{\mathsf{End}}

\newcommand{\Expect}{\mathbf{E}}
\newcommand{\Nat}{\mathbb{N}}

\newcommand{\I}{\mathcal{I}}
\newcommand{\F}{\mathcal{F}}
\newcommand{\D}{\mathcal{D}}

\renewcommand{\O}{\mathcal{O}}

\newcommand{\M}{\mathcal{M}}

\newlength{\boxwidth}
\newcommand{\TODO}[1]{\medskip \noindent
\colorbox{red!20}{
\parbox{0.98\textwidth}{
[TODO: #1]
}}
}

\newcommand{\Nnote}[1]{{\color{purple}[\textbf{Noam}: #1]}}
\newcommand{\Anote}[1]{{\color{red}[\textbf{Andrew}: #1]}}
\newcommand{\Rnote}[1]{{\color{purple}[\textbf{Rafael}: #1]}}

\author{Noam Mazor \thanks{New York University. Research supported by the Simons Foundation. }\and Andrew Morgan \thanks{Cornell Tech. Supported by AFOSR Award FA9550-24-1-0267.} \and Rafael Pass \thanks{Cornell Tech, Technion and Tel-Aviv University. Supported in part by AFOSR Award FA9550-24-1-0267, ISF Award 2338/23 and ERC Advanced Grant KolmoCrypt - 101142322. Any opinions, findings and conclusions or recommendations expressed in this material are those of the author(s) and do not necessarily reflect the views of the United States Government, the AFOSR, the European Union or the European Research Council Executive Agency.}}

\title{Language Modeling is Monotone Compression}

\begin{document}

\maketitle

\begin{abstract}
A long-standing hypothesis in artificial intelligence and neuroscience posits that intelligence is closely related to compression: the ability to compress information efficiently intuitively reflects capacities associated with intelligence and learning.
Indeed, recent experimental works verify this intuition by showing connections between the capabilities of large language models (LLMs) and their ability as compressors: for instance, Del\'etang et al. (ICLR’24) demonstrate that LLMs can be used as powerful compressors, and Huang et al. (COLM’24) show that the compression ability of LLMs is highly correlated with their performance on benchmarks for knowledge and reasoning. 

In this work, we initiate a \emph{theoretical} study of this connection. Our main result is that LLMs (formally modeled as next-token predictors) are equivalent to 
\emph{monotone} (a.k.a. order-preserving) compression algorithms---namely, compression algorithms where the encoding process preserves the ordering of the inputs---in the sense that the one can be constructed from the other while preserving the same error up to an additive gap of 2. 

We next show that the monotonicity is \emph{required} for this equivalence to hold if and only if cryptographic (infinitely-often) one-way functions exist.

%
%
%
As a direct corollary, we get a cryptographic result of independent interest: the notion of next-bit pseudoentropy (a computational analogue of entropy) of a distribution is equivalent to \emph{monotone} incompressibility of the distribution. (Previously, it was only known (Haitner et al., ITCS'23) that 
incompressibility \emph{implies} next-bit pseudoentropy.)

\end{abstract}

\section{Introduction}
A central question in the study of artificial intelligence is: can ``intelligence" be defined? Turing's classic ``imitation game" \cite{turing}, which posits that a machine should be considered intelligent if a human is unable to distinguish it from another human given a transcript of a conversation, gives an elegant answer, but one that depends on human judgment.
Towards a more mathematical and quantitative definition of intelligence, the work of Shannon \cite{shannon, shannon2} laid the groundwork for a line of research suggesting that learning or understanding something is closely related to being able to process and describe it in a simple way \cite{mahoney,hernandez}---that is, intelligence, intuitively, corresponds to the ability to succinctly \emph{compress} information.

Recently, \emph{autoregressive} models such as large language models (LLMs) 
have made this suggested connection 
more explicit. 
Such models are trained on data sampled from a distribution $D$ to \emph{predict} the next token in a sequence: given a prefix of tokens $t_1, t_2, \ldots, t_k$, they output a probability distribution that is supposed to approximate the conditional distribution $\braces{t_{k+1}\mid t_1,\ldots,t_k}$.
A model is considered ``good'' if its \emph{loss}, measured as the expected negative log-likelihood on data drawn from $D$, or equivalently, the cross-entropy between $D$ and the model’s predicted distribution, is close to minimal.

It is well known that good compression schemes can be constructed from good autoregressive models by leveraging \emph{arithmetic encoding} \cite{shannon,arithmetic}: arithmetic encoding is known to produce a near-optimal compression rate for a given distribution, so, intuitively, applying it to the distribution of a model's autoregressively sampled outputs should likewise give a near-optimal compression rate for $D$ 
if the model is good. This was first empirically studied using neural network predictors \cite{schmidhuber} (see also, e.g., \cite{schmidhuber2,mahoney2,schiopu}), but subsequently has been applied to LLMs, with recent research suggesting not only that compression rivaling state-of-the-art algorithms can be constructed from LLMs (e.g., \cite{deletang, zhang, tsai}), but also that compression ability of LLMs is highly correlated with performance on benchmarks measuring knowledge, reasoning, and coding abilities \cite{huang}. 

Given that compression can be constructed from LLMs, a natural question is whether the \emph{converse} may also hold---that is, whether one can construct accurate next-token predictors from good compression algorithms. This idea was originally empirically explored in the work of \cite{frank} constructing models for text categorization from compression algorithms, and is of particular interest because it would imply an intuitive \emph{equivalence} between predictive modeling and compression. More recently, Del\'etang et al. \cite{deletang} suggest and evaluate a heuristic construction of an autoregressive next-token prediction model from compression, leading them to their titular assertion that ``language modeling is compression".

However, when one compares the results in each direction of this proposed empirical equivalence, there is a clear difference in terms of the empirical quality of results: compression algorithms constructed from LLMs exhibit competitive compression ratios, whereas  works constructing autoregressive models from compression  exhibit relatively poor predictive ability compared to state-of-art LLMs. While there are certainly plausible intuitive explanations for this disparity---for instance, the difference in computational power between state-of-the-art LLMs and state-of-the-art compression algorithms---it nonetheless raises a fundamental question:

\begin{center}
\emph{Is there a \emph{formal} connection between autoregressive models (i.e., LLMs) and compression?}
\end{center}

In this work, we address this question, initiating a theoretical study of the intuitive connection between autoregressive models and compression. Our main result shows that, intuitively, ``good" autoregressive models are equivalent to ``good" \emph{monotone} compression schemes. 
(Furthermore, this equivalence holds quantitatively up to a constant additive accuracy gap between the cross-entropy loss of the autoregressive model and the compression loss of the monotone compressor.)
We additionally show that 
the monotonicity condition is \emph{necessary} if and only if cryptographic (infinitely-often) one-way function (ioOWFs) exist; this explains the disparity observed above.
\subsection{Our Results} 

We formalize LLMs as \emph{autoregressive models} which, given a string of tokens as a prefix, output a probability distribution over possible next tokens. We consider a setting where the objective of an LLM is, given a ``training" distribution $D$ (e.g., random files selected from the Internet), to learn a model $\M$ that minimizes the expected \emph{negative log-likelihood} $\Expect_{t \leftarrow D} [ -\sum_{i \in [\len(t)]} \log \M(t_i \; | \; t_{<i})]$, which is in fact equivalent to the \emph{cross-entropy} $H(D, D_\M) = H(D) + KL(D \; || \; D_\M)$ (where $H(D)$ is Shannon entropy and $KL$ represents Kullback-Leibler divergence) between $D$ and the distribution $D_\M$ of outputs autoregressively sampled from $\M$. 
As the minimum achievable cross-entropy is $H(D)$, we say that a model $\M$ is $\delta$-optimal with respect to $D$ if its cross-entropy is at most $H(D) + \delta$.

A compression scheme is specified by a pair of algorithms $(\Enc,\Dec)$
where $\Enc$ \emph{encodes} an input from a distribution $D$ into its
compressed form (which we assume throughout to be binary), and $\Dec$
is the respective inverse. We assume throughout that $\Enc$ is
\emph{prefix-free} with respect to $D$:
for any $x$ in the support of $D$, no codeword for a different input
is a prefix of $\Enc(x)$.

The goal of compression is to minimize the expected compression length
$\Expect_{x \leftarrow D}[\len(\Enc(x))]$; by Shannon's source coding
theorem \cite{shannon}, this is lower-bounded by $H(D)$, and so we say
that compression is $\delta$-optimal with respect to $D$ if its expected
compression length is at most $H(D)+\delta$ (in other words, the
\emph{compression loss} is at most $\delta$). Finally, compression is
\emph{monotone} if it preserves a lexicographic ordering $\prec$ over
inputs, induced by an efficiently computable ordering of the tokens:
that is, given inputs $x,y$, $\Enc(x)<\Enc(y)$ if and only if $x\prec y$,
where $<$ is the lexicographic ordering of binary codewords.

For all of our results, we require autoregressive models and
compression schemes to be \emph{efficiently computable}: encoding,
decoding, and computing next-token probabilities must run in
polynomial time.
Furthermore, when discussing
cross-entropy and expected compression length, we consider
them over complete samples from the distribution $D$.

\paragraph{Autoregressive modeling is monotone compression.} We first present our main result, showing the equivalence between autoregressive modeling and monotone compression:

\begin{theorem}[Theorem \ref{thm:llm-compression-arithmetic} and Theorem \ref{thm:compression-llm}, informal]
\label{thm:informal}
Given a distribution $D$ consisting of strings of tokens from an alphabet $\Sigma$, for any $\delta$:
\begin{itemize}
    \item Given an efficiently computable $\delta$-optimal autoregressive model for $D$, one can construct an efficiently computable monotone compression scheme which is $(\delta+2)$-optimal with respect to $D$.\footnote{The additive gap of 2 in this result, and later Corollary \ref{cor:informal}, arises from the use of arithmetic encoding in the transformation from autoregressive modeling to compression; some additive loss is inevitable as no compression scheme can achieve zero compression loss without strict conditions on the distribution $D$.}
    \item Given an efficiently computable monotone compression scheme which is $\delta$-optimal with respect to $D$, one can construct a $\delta$-optimal autoregressive model for $D$.
\end{itemize}
\end{theorem}

In other words, autoregressive models and monotone compressors are equivalent up to an additive gap of 2; we additionally demonstrate this equivalence through an explicit efficient transformation in both directions.

\paragraph{The monotonicity requirement.} 
While the added condition of monotonicity may seem surprising at first glance, it sheds light on the discrepancy we noted earlier between empirical results converting autoregressive models to compression and vice versa. Known constructions of compression from autoregressive models use arithmetic encoding, which is monotone; meanwhile, as monotonicity is not necessarily inherent in compression algorithms, it explains why constructions of autoregressive models from compression tend to have worse empirical performance.
%
Indeed, our next result formally demonstrates that the requirement of monotonicity is not merely coincidental: 
given a fundamental and widely conjectured cryptographic assumption, the existence of infinitely-often one-way functions (ioOWFs)\footnote{\emph{Infinitely-often} OWFs are a relaxation of the standard concept of a OWF: namely, OWFs which are only uninvertible for infinitely many input lengths $n$ rather than all sufficiently large $n$.}, monotonicity is \emph{necessary} for any strong autoregressive model-compression equivalence:

\begin{theorem}[Theorem \ref{thm:owf-separation} and Theorem \ref{thm:owf-inverter-to-predictor-2}, informal]
\label{thm:informal-2}
The following are true:
\begin{itemize}
\item If one-way functions (respectively, infinitely-often one-way
functions) exist, then there exists an efficiently samplable
distribution family $D=\{D_n\}$ and a non-monotone compression
scheme $(\Enc,\Dec)$ with efficient encoding and decoding whose
expected encoding length on $D_n$ is exactly $H(D_n)$, but for
which every efficiently computable autoregressive model has
excess cross-entropy loss at least linear in the output length.
This lower bound holds for all sufficiently large lengths
(respectively, infinitely many lengths). 

\item If infinitely-often one-way functions do not exist, then, given a $\delta$-optimal compression scheme $(\Enc,\Dec)$ with efficient encoding and decoding for an efficiently samplable distribution family $D=\{D_n\}$, one can construct an efficiently computable autoregressive model which is statistically close to a $\delta$-optimal autoregressive model for $\D$.
\end{itemize}
\end{theorem}

That is, given ioOWFs, one could not hope to prove Theorem \ref{thm:informal}, or even any weaker analogue with an additive gap constant sublinear in the length of a sample from the distribution $D$, without requiring monotonic compression. Interestingly, the same is not true if ioOWFs do not exist: in that case, using a classic result from \cite{il89}, we show that one can drop the monotonicity condition from Theorem \ref{thm:informal} at a cost of a small amount of statistical error. Thus, in essence, whether modeling is equivalent to \emph{monotone} compression or (approximately) equivalent to \emph{all} compression depends on the truth of this assumption; that said, one-way functions (which in turn imply ioOWFs) are widely conjectured to exist, implying the former equivalence and a corresponding separation between monotone and non-monotone compression.

\paragraph{Pseudoentropy and incompressibility.} Lastly, we show a corollary of Theorem \ref{thm:informal} that provides interesting implications to cryptography. There has been a considerable line of research relating the compressibility of a distribution to its \emph{pseudoentropy} \cite{hill}, a computational analogue of entropy.
In particular, Haitner et al. \cite{haitner} demonstrate that, if a family of distributions $\D$ is \emph{incompressible} \cite{yao} to some length $k(n)$, then it has \emph{next-bit pseudoentropy} at least $(k(n)-2)$---that is, the total pseudoentropy of the conditional next-bit distributions of each bit of $\D$ is at least $(k(n)-2)$.\footnote{We note that \cite{hirahara} demonstrates an analogue of this result for \emph{uniform} algorithms; we, however, focus on non-uniform algorithms, following \cite{haitner}.} However, the converse is false: if one-way functions exist, then pseudoentropy cannot imply incompressibility.
As our final result, we obtain a tight version of this characterization by combining Theorem \ref{thm:informal} with a corollary from \cite{vadhan} to show that, similarly to the equivalence between autoregressive models and compression, both directions are true if we specifically require \emph{monotone} incompressibility, or inability to be compressed through specifically monotone schemes:

\begin{corollary}[Corollary \ref{cor:incompressibility-pseudoentropy} and Corollary \ref{cor:pseudoentropy-incompressibility}, informal]
\label{cor:informal}
Given a family of distributions $\D$, the following are equivalent for any function $k(n)$ (to within an additive gap of 2):
\begin{itemize}
    \item $\D$ cannot be efficiently \emph{monotonically} compressed to length $k(n)$.
    \item $\D$ has next-bit pseudoentropy at least $k(n)$.
\end{itemize}
\end{corollary}

\subsection{Technical Overview}
\label{sec:overview}

\paragraph{Monotone Compression from Autoregressive Models.}
To construct monotone compression for a distribution $D$ from an autoregressive model $\M$, we
can use arithmetic encoding relative to the model's output distribution
$D_\M$. As discussed above, the expected negative log-likelihood used to
train $\M$ is exactly the cross-entropy $H(D,D_\M)$: multiplying the
next-token probabilities gives the probability of a complete output,
and taking its negative logarithm gives the sum of the next-token
losses. Arithmetic encoding connects this same quantity to compression
by assigning each output $x$ a codeword $\Enc(x)$ such that
\[
    \len(\Enc(x)) \leq -\log D_\M(x)+2.
\]
Since this bound holds for each $x$, we can average it over the correct
distribution $D$, even though the code is constructed using the approximation given by $D_\M$:
\[
    \Expect_{x \leftarrow D}[\len(\Enc(x))]
    \leq \Expect_{x \leftarrow D}[-\log D_\M(x)]+2
    =H(D,D_\M)+2.
\]
Furthermore, as mentioned above,
\[
    KL(D \; || \; D_\M)
    =\Expect_{x \leftarrow D}
        \left[\log\frac{D(x)}{D_\M(x)}\right]
    =H(D,D_\M)-H(D),
\]
so
\[
    \Expect_{x \leftarrow D}[\len(\Enc(x))]
    \leq H(D)+KL(D \; || \; D_\M)+2.
\]
In other words, the expected compression length is at most the entropy
of $D$, plus the same KL-divergence that measures
the model's excess prediction loss, plus the overhead of 2 bits from arithmetic encoding. $\delta$-optimality of $\M$ thus implies
$(\delta+2)$-optimality of the compression scheme.

Arithmetic encoding is monotone by definition: it arranges the
probability intervals associated with outputs in lexicographic order
and assigns codewords in that same order.

Efficiency of arithmetic encoding, however, requires that the cumulative distribution function (CDF) of the
model's output distribution can be computed efficiently. For a polynomial-size alphabet, this is possible using a standard recursive algorithm
(see, e.g., \cite{haitner,deletang}).
Specifically, to determine the
probability $\Pr_{x \leftarrow D_\M} [x \prec y]$ of an output being lexicographically smaller than a given
string $y$, one can decompose this into
\[ \Pr_{x \leftarrow D_\M} [x \prec y]  = \sum_i \left( \Pr_{x\gets D_m}[x_{<i} = y_{<i}]  \Pr_{x\gets D_m}[x_i \prec y_i \mid x_{<i} = y_{<i}] \right)\]
\[= \sum_i \left( \Pr_{x\gets D_m}[x_{<i} = y_{<i}] \sum_{z \prec y_i} \Pr_{x\gets D_m}[x_i = z \mid x_{<i} = y_{<i}] \right) \]
and thus recursively determine the overall probability by, for each possible prefix length where $x$ might first diverge from $y$, multiplying the probability of matching the
preceding prefix by the conditional probability of choosing a smaller
token at that position, and summing these disjoint possibilities. Adding $D_\M(y)$ then gives
the CDF $\Pr_{x \leftarrow D_\M} [x \preceq y]$. This process requires only polynomially many next-token
probability evaluations; thus, as long as the next-token probabilities
of $\M$ are efficiently computable, arithmetic encoding with respect
to $D_\M$ is as well.

\paragraph{Autoregressive Models From Monotone Compression.}
For the converse, given a prefix-free, monotone encoding $\Enc$ over
a distribution $D$, we demonstrate a construction of an autoregressive model $\M$ whose output
distribution satisfies
\[
    D_\M(x) \geq 2^{-\len(\Enc(x))}
\]
for every $x$ in the support of $D$. 
This probability bound suffices because taking negative logarithms gives
\[
    -\log D_\M(x)\leq \len(\Enc(x)).
\]
Thus, the model's loss on each output is at most its encoding length, and taking the expectation over $D$ gives a bound on the cross-entropy:
\[
    H(D,D_\M)
    =\Expect_{x\leftarrow D}[-\log D_\M(x)]
    \leq \Expect_{x\leftarrow D}[\len(\Enc(x))].
\]
In particular, if $\Enc$ is $\delta$-optimal, the right-hand side is
at most $H(D)+\delta$, so $\M$ is also $\delta$-optimal.

We proceed in
two steps. First, we construct an efficient \emph{monotone} sampler $\Sample$ for the
desired output distribution which, at a high level, samples outputs by \emph{decoding uniformly
random bits}. Let $\ell$ be a polynomial upper bound\footnote{For any $\Enc$, there must exist such an $\ell$ due to efficiency; here, we make the simplifying assumption that $\ell$ is known a priori, but we show in the complete proof that this assumption can be discarded.} on the output length of $\Enc$: $\Sample$ will draw a
uniformly random string $c \in \{0,1\}^\ell$, 
examine the prefixes of $c$ in increasing order of length, and decode the first one that is a valid codeword. Since $\Enc$ is monotone, we can always decode efficiently through binary search, and, as there are only polynomially many prefixes
to examine, this step is efficient. 
If no prefix of $c$ is a valid codeword, $\Sample$ instead returns the smallest input
$x$ with $\Enc(x)\leq c$ (or the smallest possible input $x$ if no such $x$ exists); monotonicity of $\Enc$ allows us to again find
such an input efficiently by binary search.
We observe that this process assigns an output to every random string and, importantly,
is \emph{weakly monotone}: if $c<c'$, then
$\Sample(c)\preceq\Sample(c')$. 

Next, we construct
an efficiently computable autoregressive model $\M$ whose output
distribution $D_\M$ is exactly the distribution of $\Sample(c)$ when
$c\leftarrow\{0,1\}^\ell$ is uniform.
To show that $\M$ satisfies the necessary probability bound, fix any $x$ in the support of $D$; whenever the random string $c$ begins with
$\Enc(x)$, the sampler must return exactly $x$, since prefix-freeness ensures that no
shorter prefix of $c$ can be a valid codeword. This event has
probability exactly $2^{-\len(\Enc(x))}$, so
\[
    D_\M(x)
    =\Pr_{c \leftarrow \{0,1\}^\ell}[\Sample(c)=x]
    \geq 2^{-\len(\Enc(x))}.
\]
Thus, by our original argument, $\M$ is $\delta$-optimal as long as $\Enc$ is.

However, we must still show that $\M$'s next-token probabilities can be computed efficiently.
Towards this, the key observation is that it suffices to show that the CDF of $D_\M$ can be computed efficiently, which, leveraging weak monotonicity of $\Sample$, can once again be done by binary search.
Specifically, fix an output $y$ and identify the $\ell$-bit random strings
with the integers $0,\ldots,2^\ell-1$. By weak monotonicity, the random
strings $c$ satisfying $\Sample(c)\preceq y$ correspond to a subset $c < k(y)$ of this set, and we can determine $k(y)$ using binary search by finding the smallest $c$ with $\Sample(c)\succ y$.
Since the corresponding subset contains $k(y)$ elements, each of which is
chosen with probability exactly $2^{-\ell}$, the CDF is exactly
\[
    F(y)
    :=\Pr_{x\leftarrow D_\M}[x\preceq y]
    =\frac{k(y)}{2^\ell}.
\]
Finally, to recover next-token probabilities from the CDF, we use the
standard observation that strings sharing a prefix form a
contiguous interval in a lexicographic ordering; thus, prefix
probabilities can be computed by subtracting two CDF values (see Figure \ref{fig:monotone-prefix-cdf} for a visualization),
and next-token probabilities can be computed by taking ratios of prefix
probabilities. We have thus demonstrated that the model $\M$ we constructed has both efficiently computable next-token probabilities and $\delta$-optimality with respect to $D$, as desired.
\begin{figure}[t]
  \centering
  \begingroup
  \definecolor{samplerink}{HTML}{24282C}
  \definecolor{samplermuted}{HTML}{58616A}
  \definecolor{sampleraccent}{HTML}{245F83}
  \definecolor{samplerband}{HTML}{EAF1F5}
  \definecolor{samplerprojection}{HTML}{F3F7FA}
  \definecolor{sampleroverlap}{HTML}{D6E6F0}
  \definecolor{samplerguide}{HTML}{7897AA}
  \begin{tikzpicture}[x=0.82pt,y=0.82pt,>=stealth]
    \path[use as bounding box] (0,0) rectangle (492,286);
    \fill[samplerband] (60,147.5) rectangle (470,210.5);
    \fill[samplerprojection] (187.5,74) rectangle (352,147.5);
    \fill[sampleroverlap] (187.5,147.5) rectangle (352,210.5);
    \draw[draw=samplerguide,line width=0.533pt,line join=round,dash pattern=on 2.624pt off 2.624pt] (60,147.5) -- (470,147.5);
    \draw[draw=samplerguide,line width=0.533pt,line join=round,dash pattern=on 2.624pt off 2.624pt] (60,210.5) -- (470,210.5);
    \draw[draw=sampleraccent,line width=0.656pt,line join=round,dash pattern=on 2.624pt off 2.624pt] (187.5,74) -- (187.5,147.5);
    \draw[draw=sampleraccent,line width=0.656pt,line join=round,dash pattern=on 2.624pt off 2.624pt] (352,74) -- (352,210.5);
    \draw[draw=samplerink,line width=0.82pt,line join=round,->] (60,74) -- (479,74);
    \draw[draw=samplerink,line width=0.82pt,line join=round,->] (60,74) -- (60,274);
    \node[anchor=center,inner sep=0pt,text=samplerink,rotate=90,font={\fontsize{9.02}{10.824}\selectfont}] at (-8,178) {Outputs (lexicographic order)};
    \node[anchor=center,inner sep=0pt,text=samplerink,rotate=0,font={\fontsize{9.02}{10.824}\selectfont}] at (270,15) {Random strings $c$ (binary order)};
    \draw[draw=samplerink,line width=1.394pt,line join=round] (70,95) -- (117,95) -- (117,116) -- (164,116) -- (164,137) -- (187.5,137) -- (187.5,158) -- (211,158) -- (211,179) -- (258,179) -- (258,200) -- (352,200) -- (352,221) -- (399,221) -- (399,242) -- (446,242);
    \draw[draw=sampleraccent,line width=1.845pt,line join=round] (187.5,147.5) -- (187.5,158) -- (211,158) -- (211,179) -- (258,179) -- (258,200) -- (352,200) -- (352,210.5);
    \node[anchor=center,inner sep=0pt,text=samplerink,rotate=0,font={\fontsize{9.84}{11.808}\selectfont}] at (402,263) {$\Sample(c)$};
    \draw[draw=sampleraccent,line width=0.656pt,line join=round] (56,158) -- (60,158);
    \node[anchor=east,inner sep=0pt,text=sampleraccent,rotate=0,font={\fontsize{7.38}{8.856}\selectfont\itshape}] at (52,158) {$t_{<i} ||\sigma_1\ldots$};
    \draw[draw=sampleraccent,line width=0.656pt,line join=round] (56,179) -- (60,179);
    \node[anchor=east,inner sep=0pt,text=sampleraccent,rotate=0,font={\fontsize{7.38}{8.856}\selectfont\itshape}] at (52,179) {$t_{<i} || \sigma_2\ldots$};
    \draw[draw=sampleraccent,line width=0.656pt,line join=round] (56,200) -- (60,200);
    \node[anchor=east,inner sep=0pt,text=sampleraccent,rotate=0,font={\fontsize{7.38}{8.856}\selectfont\itshape}] at (52,200) {$t_{<i} || \sigma_3\ldots$};
    \node[anchor=center,inner sep=0pt,text=sampleraccent,rotate=0,font={\fontsize{9.02}{10.824}\selectfont}] at (134,186) {outputs beginning};
    \node[anchor=center,inner sep=0pt,text=sampleraccent,rotate=0,font={\fontsize{9.02}{10.824}\selectfont}] at (134,172) {with prefix $t_{<i}$};
    \draw[draw=sampleraccent,line width=2.05pt,line join=round] (187.5,74) -- (352,74);
    \draw[draw=sampleraccent,fill=sampleraccent,line width=0.984pt] (187.5,74) circle[radius=1.968pt];
    \draw[draw=sampleraccent,fill=white,line width=0.984pt] (352,74) circle[radius=1.968pt];
    \draw[draw=sampleraccent,line width=0.656pt,line join=round] (187.5,67) -- (187.5,61) -- (352,61) -- (352,67);
    \node[anchor=center,inner sep=0pt,text=sampleraccent,rotate=0,font={\fontsize{9.02}{10.824}\selectfont}] at (269.75,49) {strings $c$ producing prefix $t_{<i}$};
    \node[anchor=center,inner sep=0pt,text=samplermuted,rotate=0,font={\fontsize{7.38}{8.856}\selectfont}] at (269.75,33) {endpoints found by binary search};
  \end{tikzpicture}
  \endgroup
  \caption{Computing prefix probabilities from a monotone sampler.
  Outputs beginning with a fixed prefix $t_{<i}$ form an interval in
  lexicographic order (horizontal band). The labels $t_{<i} || \sigma_1\ldots$,
  $t_{<i} || \sigma_2\ldots$, and $t_{<i} || \sigma_3\ldots$ illustrate outputs with this common
  prefix, where $\sigma_1\prec \sigma_2\prec \sigma_3$ are possible next tokens.
  By weak monotonicity of
  $\Sample$, the random strings producing these outputs also form
  an interval (highlighted on the horizontal axis). Two binary
  searches locate its endpoints. The prefix probability is the
  number of seeds in this interval divided by $2^\ell$. The staircase is schematic;
  the vertical axis represents output order.}
  \label{fig:monotone-prefix-cdf}
\end{figure}
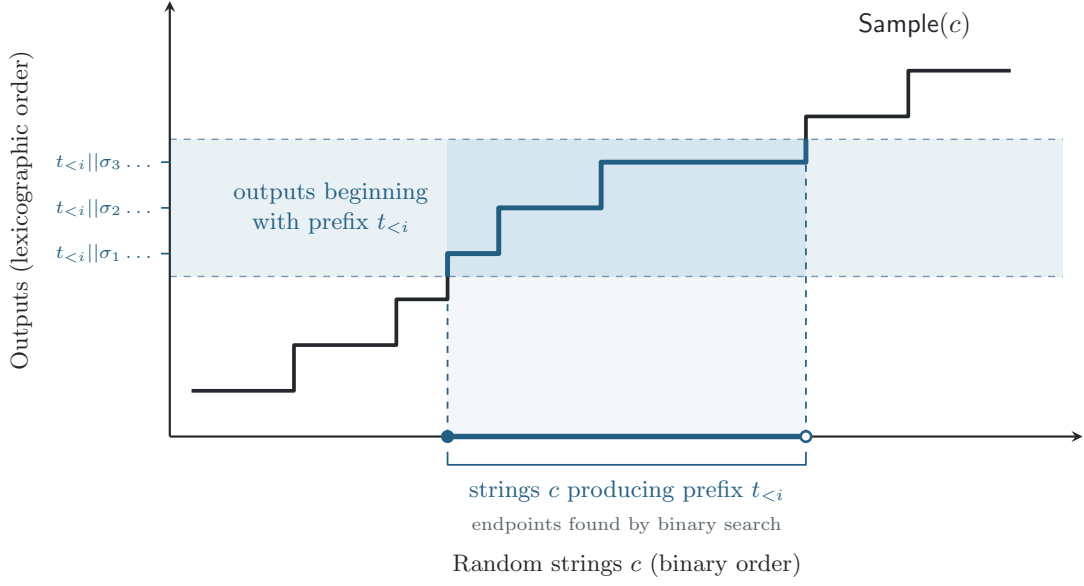

\paragraph{Comparison with \cite{deletang}.} Del\'etang et al. \cite{deletang} present a heuristic construction of next-token prediction from compression; however, in contrast to our construction, which ``looks ahead" and retrieves next-token probabilities relative to the \emph{complete output distribution} reconstructed from the compression, they instead compute next-token probabilities myopically by examining the difference in length between the compression of the prefix and the compression of the prefix with each possible subsequent token appended.\footnote{That is, given some partial output $x$, they assign each possible next token $\sigma$ a weight of $2^{-(\len(\Enc(x || \sigma))-\len(\Enc(x)))}$ and sample according to the distribution over all possible next tokens $\sigma$ with the respective weights.} The advantage of this construction is that, unlike ours, monotonicity of the compression scheme $\Enc$ is not needed. However (as the authors themselves acknowledge in their work), because their approach approximates the next-token distribution using partial compressions every time they sample a next token, it risks introducing a constant amount of error to the sampling process every time a token is generated, whereas, with our construction, using a consistent distribution over complete outputs allows us to sample all tokens relative to that output distribution and only incur the error once.\footnote{As an illustrative example, consider a binary distribution $D$ with constant next-token probabilities skewed away from powers of 2: e.g., every token is drawn as 0 with probability 0.6 and 1 with probability 0.4. Since every token is drawn with probabilities weighted by powers of 2 in the sampling method of \cite{deletang}, this introduces a constant amount of KL-divergence to the respective next-token distribution every time a token is sampled; as these errors are additive by the chain rule, the resulting cross-entropy loss is linear in the output length.}

\section{Definitions}
\label{sec:definitions}

Let $\Nat = \braces{1, 2, 3, \ldots}$ and $[n] = \braces{m \in \Nat \mid m \leq n} = \braces{1, 2, \ldots, n}$. For some symbol $\sigma$ in an alphabet $\Sigma$, take $\sigma^n$ to be the string consisting of $n$ copies of $\sigma$. Let $\textbf{1}_{P}$ be the \emph{indicator function} of a logical statement $P$, defined as being 1 when $P$ is true and 0 otherwise. Given a set $X$, we define $\Delta(X)$ to be the set of probability distributions over elements in $X$. Given some string $x \in \Sigma^*$ over a vocabulary $\Sigma$, we will denote substrings with notation such as $x_{\leq i}$ (representing the prefix consisting of the first $i$ symbols), and the length of a string $x \in \Sigma^*$ by $\len(x)$. We use similar notation to denote, e.g., $\Sigma^{\leq n} = \bigcup_{i = 0}^n \Sigma^i$ as the set of strings of tokens from an alphabet $\Sigma$ with at most length $n$.

We say a statement is true for ``sufficiently large $n \in \Nat$" if there exists $N \in \Nat$ such that the statement holds for all $n \geq N$. We shall call a function $\epsilon(\cdot)$ \emph{negligible} if, for any polynomial $p(\cdot)$, $\epsilon(n) < 1/p(n)$ for sufficiently large $n \in N$; conversely, a function is \emph{non-negligible} or \emph{polynomial} if the above is not true.

In this work, we assume that algorithms are \emph{non-uniform}: to enable computation for arbitrary input sizes parameterized by some $n \in \Nat$, they receive an additional ``advice" string of size polynomial in $n$ (which, unless otherwise specified, can be taken to be the string $1^n$ of $n$ ones to specify $n$). We refer to a non-uniform algorithm as \emph{efficient} if its running time is polynomial in the parameter $n$. We define an \emph{oracle-aided algorithm} $\Pi^{\O}$ as an algorithm $\Pi$ that is given access to another algorithm $\O$ as a black box, and may make many calls to its oracle $\O$ on potentially many different inputs before determining its own output. We say that such an algorithm is efficient if its running time (excepting calls made to $\O$) and the number of calls it makes to the oracle $\O$ are both polynomial in its input size---that is, as long as $\O$ is efficient, the total running time of $\Pi$, including that taken by calls made to $\O$, is polynomial in its input size.

Unless otherwise indicated, any logarithm denoted by $\log$ will be base 2 (as opposed to the base-$e$ natural logarithm $\ln$). Given a discrete distribution $D$, we denote by $\sup(D)$ its support and define the (Shannon) \emph{entropy} of $D$ to be the quantity
\[H(D) = \sum_{x \in \sup(D)} -D(x)\log D(x)\]

Given two discrete distributions $D, D'$, we denote the \emph{statistical distance} (or total variation distance) between $D$ and $D'$ by $\Delta(D, D')$ and the Kullback-Leibler (KL) divergence between the two distributions by
\[ KL(D || D') = \sum_{x \in \sup(D)} D(x) \log \frac{D(x)}{D'(x)} \]


\subsection{Compression}

We consider the problem of (lossless) data compression from a particular source distribution, where the goal is to achieve a uniquely decodable encoding over samples from the source distribution with the property that the expected encoding length of a sample drawn from the distribution is as small as possible. Formally:

\begin{definition}[\cite{gs85,TrevisanVZ05}]
\label{def:compression}
Given finite families of alphabets $\Sigma = \braces{\Sigma_n}_{n \in \Nat}$ and $\Sigma' = \braces{\Sigma'_n}_{n \in \Nat}$ where $|\Sigma_n|$ and $|\Sigma'_n|$ are polynomial in $n$, functions $m(\cdot), \ell(\cdot)$ where $m(n)$ is polynomial in $n$, and an ensemble of distributions $\D = \braces{D_n}_{n \in \Nat}$, where $\sup(D_n) \subseteq \Sigma_n^{m(n)}$, we say that a family of injective functions $\Enc = \braces{\Enc_n}_{n \in \Nat}$ is an \textbf{encoding}, or \textbf{compression scheme}\footnote{Oftentimes the term ``compression" is used to indicate strictly prefix-free encodings; however, we consider only prefix-free encodings anyway in our results.}. $\Enc$ \textbf{compresses $\D$ to length} $\ell(\cdot)$ if, for any $n \in \Nat$, $\Enc_n : \Sigma_n^{m(n)} \rightarrow \Sigma'^*_n$ is such that
\[ \Expect_{x \leftarrow D_n} [\len(\Enc_n(x))] \leq \ell(n) \]
\end{definition}
Unless otherwise specified, we shall assume that the output of the compression is binary ($\Sigma'_n = \Bit$ for all $n \in \Nat$). We also define some useful properties of compression functions:
\begin{definition}[Properties of compression]
Given finite alphabets $\Sigma$ and $\Sigma'$, let $\Enc : \Sigma^* \rightarrow \Sigma'^*$ be a function defined with respect to some distribution $D$ with $\sup(D) \subseteq \Sigma^*$.

We refer to $\Enc$ as \textbf{prefix-free} if there do not exist elements $x \neq y$ and $i \in \Nat$ such that $x \in \Sigma^*$, $y \in \sup(D)$, and $\Enc(x) = (\Enc(y))_{< i}$ (i.e., $\Enc(x)$ is a prefix of $\Enc(y)$; we note, however, that this definition does allow an encoding of $y \in \sup(D)$ to be a prefix of an encoding of $x \not\in \sup(D)$).


We refer to $\Enc$ as \textbf{monotone} if there exist efficiently computable orderings $\prec_\Sigma$ over $\Sigma$ and $\prec_{\Sigma'}$ over $\Sigma'$ such that, for any $x, y \in \Sigma^*$, $x \prec_{\Sigma^*} y$ if and only if $\Enc(x) \prec_{\Sigma'^*} \Enc(y)$, where $\prec_{\Sigma^*}$ is the lexicographic ordering\footnote{i.e., $x \prec_{\Sigma^*} y$ when there exists $i \in \Nat$ for which $x_{<i} = y_{<i}$ and either $\len(y) > \len(x) = i-1$ or $x_i \prec_{\Sigma} y_i$.} induced by $\prec_\Sigma$ on $\Sigma^*$ and respectively for $\prec_{\Sigma'^*}$. 

We consider any of the above properties to apply to a family of functions $\Enc = \braces{\Enc_n}_{n \in \Nat}$ as defined in Definition \ref{def:compression}, where $\Enc_n$ is defined with respect to a distribution $D_n$ for which $\sup(D_n) \subseteq \Sigma_n^{m(n)}$, if they apply to $\Enc_n$ and $D_n$ for each $n \in \Nat$.

Furthermore, we say that such a family of functions has \textbf{efficient encoding} if there exists a polynomial $p(\cdot)$ such that $\Enc_n$ runs in time at most $p(n)$ for any $n \in \Nat$, and \textbf{efficient decoding} if there exists a polynomial $p'(\cdot)$ and a function family $\Dec = \braces{\Dec_n}_{n \in \Nat}$ such that, for each $n \in \Nat$, $\Dec_n$ computes the inverse $\Enc_n^{-1}(c)$ of $\Enc_n$ for any codeword in its image (or returns $\bot$ if the input is not in the image) and additionally runs in time at most $p'(n)$.
\end{definition}



We also note the following useful fact about efficient encoding and decoding:
\begin{fact}
\label{fact:decoding}
Any monotone and efficiently encodable family of encodings $\Enc = \braces{\Enc_n}_{n \in \Nat}$ is also efficiently decodable.
\end{fact}
\begin{proof}
Consider a decoding algorithm that, given a codeword, performs binary search over the possible inputs in $\Sigma_n^{m(n)}$. This is possible due to monotonicity---it can be determined whether the correct decoding is lower or higher than a given input by comparing the respective codewords. Additionally, since there are $|\Sigma_n|^{m(n)}$ possible inputs in $\Sigma_n^{m(n)}$, binary search requires at most $O(m(n)\log |\Sigma_n|)$ iterations of $\Enc_n$. Since by definition $m(n)$ and $|\Sigma_n|$ are polynomial in $n$, and since $\Enc$ is itself efficient, this decoding process must likewise be efficient.
\end{proof}

We shall generally restrict our focus to monotone and prefix-free encodings. Shannon's source coding theorem \cite{shannon} implies a lower bound for the expected length of any prefix-free encoding $\Enc$ over a source distribution $D$ in terms of its Shannon entropy $H(D)$:

\begin{corollary}[\cite{shannon}]
\label{cor:shannon}
Given finite alphabets $\Sigma$ and $\Sigma'$ and distribution $D$ over $\Sigma$, let $\ell^*(D)$ be the minimal expected length $\Expect_{x \leftarrow D} [\len(\Enc(x))]$ over any prefix-free encoding $\Enc : \Sigma \rightarrow \Sigma'$. Then:
\[ \frac{H(D)}{\log |\Sigma'|} \leq \ell^*(D) \]
\end{corollary}
For binary encodings ($\Sigma' = \Bit$), this bound simply becomes $\ell^*(D) \geq H(D)$.

\subsubsection{Arithmetic Encoding}

Towards constructing compression with the properties that are necessary for our results, we start with the classical notion of arithmetic encoding, which essentially maps each element of a distribution $D$ to a binary representation of the element's cumulative distribution function with respect to $D$ with sufficient precision to be unique.

\begin{definition}[Arithmetic encoding \cite{shannon,arithmetic}]\label{def:arithmetic_encoding}
Let $D$ be a discrete distribution and let $\prec$ be a total order over $\sup(D)$. Let $F : \sup(D) \rightarrow [0,1]$ by $F(x)=(\sum_{a\prec x}D(a))+1/2\cdot D(x)$. Then the \textbf{arithmetic encoding} $\Enc : \sup(D) \rightarrow \Bit^*$ is defined such that, for $x \in \sup(D)$, $\Enc(x)$ is the first $(\ceil{-\log D(x)} +1)$ bits of $F(x)$. 	
\end{definition}

We note the following fact that shows that, \emph{when restricted to the support of the distribution} $D$, arithmetic encoding is prefix-free and monotone, and additionally provides near-optimal data compression:

\begin{fact}[Properties of arithmetic encoding, lemma 2.8 in \cite{TrevisanVZ05}]\label{fact:arithmetic_encoding}
For every distribution family $\D = \braces{D_n}_{n \in \Nat}$ and ordering $\prec$ over the corresponding input alphabet $\Sigma_n$, the functions $\Enc_n$ given by applying Definition \ref{def:arithmetic_encoding} to $D_n$ and $\prec$ (for any $n \in \Nat$) are one-to-one and monotone, and the scheme $\Enc = \braces{\Enc_n}_{n \in \Nat}$ is prefix-free and compresses $\D$ to length $H(D_n) + 2$.

\end{fact}

This is not quite what we need to show monotonicity and prefix-freeness with respect to the entire input space $\Sigma^{m(n)}_n$---in fact, it is not necessarily clear how to define $\Enc$ outside the support of $\D$ since $(\ceil{-\log D_n(x)} +1)$ (the number of bits in $\Enc_n(x)$) is undefined when $D_n(x) = 0$. However, we show that one can transform arithmetic encoding to guarantee the stronger versions of these properties, and additionally show that arithmetic encoding, including our augmented version, is efficiently computable as long as the \emph{cumulative distribution function} (CDF) $\sum_{a\preceq x}D_n(a)$ is likewise efficiently computable for any $x \in \Sigma^{m(n)}_n$:

\begin{claim*}
\label{clm:arithmetic-properties}
For any alphabet $\braces{\Sigma_n}_{n \in \Nat}$, any distribution family $\D = \braces{D_n}_{n \in \Nat}$ where $\sup(D_n) \subseteq \Sigma_n^{m(n)}$ and $-\log(D_n(x)) \leq \ell(n)$ for any $x \in \sup(D_n)$ and some polynomial $\ell(\cdot)$, and any ordering $\prec_n$ over $\Sigma_n$, if the cumulative distribution function $\sum_{a\preceq x}D_n(a)$ can be efficiently computed for any $x \in \Sigma^{m(n)}_n$, then there exists a prefix-free, monotone, and efficiently computable encoding $\Enc = \braces{\Enc_n}_{n \in \Nat}$ that additionally compresses $\D$ to length $H(D_n) + 2$.
\end{claim*}

\begin{proof}
Let $\Enc'_n$ be the arithmetic encoding given in Definition \ref{def:arithmetic_encoding}. Define $\Enc_n(x) = \Enc'_n(x)$ for all $x \in \sup(D_n)$; for $x \in \Sigma^{m(n)}_n \setminus \sup(D_n)$, let $\Enc_n(x)$ be the first $(\ell(n)+2)$ bits of the CDF $\sum_{a\prec x}D(a)$ (or $1^{\ell(n)+2}$ if the CDF is 1, i.e., $x \succ y$ for all $y \in \sup(D_n)$) with the binary representation of $x$ ordered by $\prec_n$ appended.

Prefix-freeness holds because, by Fact \ref{fact:arithmetic_encoding}, $\Enc_n$ is prefix-free when restricted to $\sup(D_n)$, and furthermore there cannot exist $x \in \Sigma^{m(n)}_n \setminus \sup(D_n)$ and $y \in \sup(D_n)$ where $\Enc_n(x)$ is a prefix of $\Enc_n(y)$ because, by definition, $\Enc_n(x)$ must be strictly longer than $\Enc_n(y)$.

To show monotonicity, consider some $x, y \in \Sigma^{m(n)}_n$ with $x \prec_n y$. If $x, y \in \sup(D_n)$, we conclude that $\Enc_n(x) < \Enc_n(y)$ by Fact \ref{fact:arithmetic_encoding}. If the first $(\ell(n)+2)$ bits of $\Enc_n(x)$ and $\Enc_n(y)$ are identical (which can only happen when $x, y \not\in \sup(D_n)$ and furthermore there is no $z \in D_n$ with $x \prec_n z \prec_n y$), we conclude that $\Enc_n(x) < \Enc_n(y)$ by the definition of $\Enc_n$ (since the following part of each of $\Enc_n(x)$ and $\Enc_n(y)$ is the binary representation of $x$ and $y$, respectively). Otherwise, it must be the case that $x$ and $y$ do not have the same CDF with respect to $D_n$, and so we conclude that $\Enc_n(x) < \Enc_n(y)$ because the respective prefixes of the encodings (which are prefixes of the respective CDFs) must be strictly ordered. 

To show that $\Enc_n$ can be efficiently computed: if $\sum_{a\preceq x}D_n(a)$ can be efficiently computed, then so can 
\[ D_n(x) = \sum_{a \preceq x}D_n(a) - \sum_{a\preceq x^-}D_n(a) \]
(where $x^-$ is the predecessor of $x$ under $\prec$) and consequently
\[ \sum_{a\prec x}D_n(a) = \sum_{a\preceq x}D_n(a) - D(x)\]
and thus $\Enc'_n(x)$ must be efficiently computable as well, which demonstrates efficient computability of $\Enc_n$ for $x \in \sup(D_n)$. For $x \not\in \sup(D_n)$, it suffices that the CDF $\sum_{a\prec x}D(a)$ is efficiently computable.

Finally, we remark that the encoding length of every $x \in \sup(D_n)$ is identical between $\Enc'_n$ and $\Enc_n$; therefore, the expected encoding length with respect to $D_n$ must remain unchanged.
\end{proof}

\begin{remark}
    We note the assumption in the above claim that we have an effective \emph{a priori} bound of $p(n)+1$ on the length of the arithmetic encoding $\Enc'_n(x)$. Without such a bound, one can still prove a similar result by truncating the arithmetic encoding to a suitably large polynomial length $\ell(n)$ and appending the binary representation of $x$ to any truncated encoding, even for $x \in \sup(D_n)$ with non-zero $D_n(x)$. Since any $x \in \sup(D_n)$ must have $D_n(x) \leq 2^{-\ell(n)}$ for its encoding to be affected, and there are at most $|\Sigma_n|^{m(n)}$ possible codewords, as long as $\ell(n)$ is selected so that $2^{\ell(n)} \gg |\Sigma_n|^{m(n)}$, this will lose at most an exponentially small fraction of a bit in the output length. We, however, choose instead to assume an \emph{a priori} bound $\ell(n)$ since, in the setting we consider, any efficiently computable family of autoregressive models must produce a distribution family $\D = \braces{D_n}_{n \in \Nat}$ where the arithmetic encoding length is polynomially bounded anyway.
\end{remark}

\subsubsection{KL-divergence and Compression Loss}

One intuitive interpretation of the KL-divergence $KL(D \; || \; D')$ between two distributions is as the expected number of additional bits that would be needed to encode samples from $D$ using compression optimized for the distribution $D'$ as opposed to using compression optimized for $D$. We restate this for schemes, such as arithmetic encoding, which are ``close to optimal" in the sense that the encoding length of a sample is additively close to the empirical entropy of that sample:

\begin{claim*}
\label{clm:optimal-encoding}
Let $D, D'$ be distributions on $\Sigma^m$ and $\Enc$ be an encoding such that, for any element $x \in \sup(D')$:
\[ \len(\Enc(x)) \leq \log(1/D'(x)) + \alpha \]
Then, for any distribution $D$ (with $\sup(D) \subseteq \sup(D')$):
\[ \Expect_{x \leftarrow D}[\len(\Enc(x))]  \leq H(D) + KL(D \; || \; D') + \alpha \]
\end{claim*}

\begin{proof}
By the definition of KL-divergence:
\begin{align*}
 KL(D \; || \; D') &\; = \Expect_{x \leftarrow D} [\log (1/D'(x))] - \Expect_{x \leftarrow D} [\log (1/D(x))] \\
&\; \geq (\Expect_{x \leftarrow D} [\len(\Enc(x))] - \alpha) - H(D) \\
\end{align*}
which directly implies the conclusion.
\end{proof}

By definition, the family $\Enc = \braces{\Enc_D}$ where $\Enc_D$ is the arithmetic encoding relative to $D$ has this property for $\alpha = 2$ and thus satisfies:
\[ \Expect_{x \leftarrow D}[\len(\Enc_{D'}(x))]\leq H(D) + KL(D \; || \; D') + 2 \]

\subsection{Cross Entropy and LLM Training Objectives}

We recall the definition of cross-entropy:

\begin{definition}
The \textbf{cross-entropy} of a distribution $D_2$ relative to a distribution $D_1$ is given by:
\[ H(D_1, D_2) = - \Expect_{x \leftarrow D_1} [\log D_2(x)] = - \sum_{x \in \sup(D_1)} D_1(x) \log D_2(x) \]
\end{definition}
Cross-entropy is closely related to KL-divergence by the following fact:
\begin{fact}
\label{fact:cross-kl}
For any distributions $D_1, D_2$:
\[ H(D_1, D_2) = KL(D_1 || D_2) + H(D_1) \]
\end{fact}

\begin{proof}
    \begin{align*}
        H(D_1, D_2) &= - \sum_{x \in \sup(D_1)} D_1(x) \log D_2(x) = \sum_{x \in \sup(D_1)} D_1(x) \log \frac{D_1(x)}{D_2(x)} - \sum_{x \in \sup(D_1)} D_1(x) \log D_1(x) \\&=   KL(D_1 || D_2) + H(D_1). 
    \end{align*}
\end{proof}

\subsubsection{Autoregressive Models}

We shall generally treat LLMs in a black-box way and consider the process of training an LLM in terms of the objective function that it attempts to optimize. Let us assume that an LLM is being trained on some distribution $D$ over sequences of tokens $t \in \Sigma^*$ from a vocabulary $\Sigma$. During pretraining, the majority of the training process, an LLM is given samples from the distribution $D$ and attempts to ``learn" the distribution $D$ by building an \emph{autoregressive model} $\M$ that maps any sequence of tokens $t_{<i} \in \Sigma^*$ to a distribution $\M_{t_{<i}}$ over next tokens that assigns each possible next token $t_i$ some probability $\M(t_i \; | \; t_{<i})$. Formally:

\begin{definition}
Given an alphabet $\Sigma$ and a maximum output length $m$,\footnote{We can assume without loss of generality that outputs of length less than $m$ are padded to length $m$ by a special $\End$ token.} we define an \textbf{autoregressive model} $\M$ over $\Sigma^m$ to be a collection of distributions $\braces{\M_{t}}_{t \in \Sigma^{< m}}$ where every distribution $\M_t$ has support $\sup(\M_t) \subseteq \Sigma$. We define the induced distribution $D_\M$ to be the distribution of complete outputs resulting from the process of starting with an empty string $t$ and iteratively appending a sample from $\M_t$ to $t$ $m$ times to construct a length-$m$ string $t \in \Sigma^m$.
\end{definition}

The objective of this training process, then, is to minimize the expectation over samples $t \leftarrow D$ of the loss function
\begin{equation}
\label{eqn:loss}
L(\M, t) = -\sum_{i \in [\len(t)]} \log \M(t_i \; | \; t_{<i})
\end{equation}
which is known as \emph{cross-entropy loss}
or \emph{negative log-likelihood}. We refer to an autoregressive model as $\delta$-\emph{optimal} with respect to a distribution $D$ if, for any 
autoregressive model $\M^*$ defined over the same alphabet:
\[ \Expect_{t \leftarrow D}[L(\M, t)] - \Expect_{t \leftarrow D}[L(\M^*, t)] \leq  \delta \]
that is, the expected loss of $\M$ over $D$ is no more than $\delta$ worse than the minimum achievable loss. An important observation about this loss function is the following:
\begin{claim*}
\label{clm:loss-kl}
For alphabet $\Sigma$, maximum length $m$, and a fixed training distribution $D$ with support $\sup(D) \subseteq \Sigma^m$, if there exist an autoregressive model $\M$ over $\Sigma^m$ and constant $\delta$ such that $\M$ is $\delta$-optimal with respect to $D$,
then, for any autoregressive model $\M^*$ defined over $\Sigma^m$:
\[KL(D \; || \; D_\M) - KL(D \; || \; D_{\M^*}) \leq \delta \]
\end{claim*}

\begin{proof}
Let $m$ be the maximum output length of a sample from $D$. We can apply the chain rule and Fact \ref{fact:cross-kl} to show:
\begin{align*}
KL(D \; || \; D_\M) &\; = \sum_{i \in [m]} \Expect_{t_{<i} : t \leftarrow D} [KL(D_{t_{<i}} \; || \; \M_{t_{<i}})]  \\
&\; = \sum_{i \in [m]} \Expect_{t_{<i} : t \leftarrow D} [H(D_{t_{<i}}, \M_{t_{<i}}) - H(D_{t_{<i}})]  \\
\end{align*}
where $D_{t_{<i}}$ is defined analogously to $\M_{t_{<i}}$ (i.e., the next-token distribution given prefix $t_{<i}$) for the training distribution $D$. Furthermore:
\begin{align*}
&\;  \sum_{i \in [m]} \Expect_{t_{<i} : t \leftarrow D} [H(D_{t_{<i}}, \M_{t_{<i}})]  \\
= &\; - \sum_{i \in [m]} \Expect_{t_{<i} : t \leftarrow D} [\Expect_{t_{i} \leftarrow D_{t_{<i}}} [\log \M(t_i \; | \; t_{<i})]]  \\
= &\; - \sum_{i \in [m]} \Expect_{t_{\leq i} : t \leftarrow D} [\log \M(t_i \; | \; t_{<i})]  \\
= &\;  \Expect_{t \leftarrow D} \left[ - \sum_{i \in [m]} \log \M(t_i \; | \; t_{<i}) \right] =  \Expect_{t \leftarrow D} \left[ - \sum_{i \in [\len(t)]} \log \M(t_i \; | \; t_{<i}) \right] = \Expect_{t \leftarrow D} [L(\M, t)] \\
\end{align*}
Importantly, since the training distribution $D$ is fixed, the entropy $H(D_{t_{<i}})$ is constant, as is its expectation over $t \leftarrow D$, which provides that $KL(D \; || \; D_\M)$ and $\Expect_{t \leftarrow D} [L(\M, t)]$ are separated by the same additive constant for any model $\M$, implying the conclusion since we have 
\[ \Expect_{t \leftarrow D}[L(\M, t)] - \Expect_{t \leftarrow D}[L(\M^*, t)] \leq  \delta \]
by assumption of $\delta$-optimality.
\end{proof}

This intuitively means that training an LLM by attempting to minimize the expected loss $\Expect_{t \leftarrow D} [L(\M, t)]$ is equivalent to minimizing the KL-divergence $KL(D \; || \; D_\M)$.

One final result that will be useful to us is demonstrating that the cumulative distribution function of a distribution $D_\M$ induced by an autoregressive model $\M$ can be efficiently computed:

\begin{claim*}
\label{clm:efficient-cumulative}
Given a polynomial-sized alphabet family $\Sigma = \braces{\Sigma_n}_{n \in \Nat}$ with lexicographical ordering $\prec$, polynomial maximum length $m(\cdot)$, and family of autoregressive models $\M = \braces{\M_n}_{n \in \Nat}$ over $\Sigma_n^{m(n)}$, the following statements are equivalent:

\begin{enumerate}
    \item For any $t \in \Sigma_n^{m(n)}$, $\M_n(t_i \; | \; t_{<i})$ is efficiently computable on all prefixes with $\Pr_{x \leftarrow D_{\M_n}}[x_{< i} = t_{< i}] > 0$.
    \item For any $n \in \Nat$ and $a \in \Sigma_n^{m(n)}$, the cumulative distribution function $\Pr_{x \leftarrow D_{\M_n}}[x \preceq a]$ is efficiently computable.
\end{enumerate}
\end{claim*}

\begin{proof}

\emph{(1) $\rightarrow$ (2).} Consider the following procedure for computing $\Pr_{x \leftarrow D_{\M_n}}[x \preceq a]$:
\begin{itemize}
\item Let $p = 0$, $p' = 1$, and $i = 1$.
\item Let $M = (\M_n)_{a_{<i}}$, the next-token distribution for the $i^\text{th}$ token of $a$.
\item Set $p \leftarrow p + p' ( \sum_{t \in \text{sup}(M):t \prec a_i} M(t))$ and set $p' = p' \cdot M(a_i)$.
\item Repeat the previous two steps incrementing $i$ until $i = m(n)$ or $p' = 0$, then take $p \leftarrow p + p'$, terminate and return $p$.
\end{itemize}
Correctness holds because each iteration adds the probability that $x \leftarrow D_{\M_n}$ has $x_{<i} = a_{<i}$ but $x_i \prec a_i$ (and the final step $p \leftarrow p + p'$ adds the probability that $x = a$), and the sum of those probabilities is precisely $\Pr_{x \leftarrow D_{\M_n}}[x \preceq a]$.

Efficiency holds because, since $|\Sigma|$ is polynomial, we compute $M(t) = \M_n(t \; | \; a_{<i})$ at most a polynomial number of times in each iteration of the loop, which itself iterates a polynomial number $p(n)$ of times. 

\emph{(2) $\rightarrow$ (1).} 
This follows since 

\[\M_n(t_i \; | \; t_{<i}) = \frac{\Pr_{x \leftarrow D_{\M_n}}[x_{\leq i} = t_{\leq i}]}{\Pr_{x \leftarrow D_{\M_n}}[x_{< i} = t_{< i}]} =  \frac{\Pr_{x \leftarrow D_{\M_n}}[x \preceq t_{\leq i}^\top] - \Pr_{x \leftarrow D_{\M_n}}[x \preceq t_{\leq i}^\bot]}{\Pr_{x \leftarrow D_{\M_n}}[x \preceq t_{<i}^\top] - \Pr_{x \leftarrow D_{\M_n}}[x \preceq t_{<i}^\bot]}\]
where $t_{<i}^\top$ is the $\prec$-greatest element of $\Sigma_n^{m(n)}$ with $t_{<i}$ as a prefix and $t_{<i}^\bot$ is the $\prec$-greatest element of $\Sigma_n^{m(n)}$ with a prefix $(t^\bot)_{<i} \prec t_{<i}$, and respectively for $t_{\leq i}^\top$ and $t_{\leq i}^\bot$. 
\end{proof}

\section{Autoregressive Models and Compression}
\label{sec:equivalence}

In this section, we formally prove the equivalence (up to an additive constant) that we informally stated in Theorem \ref{thm:informal} between the notions of optimal compression of a distribution and learning an accurate autoregressive model for that distribution. Using the facts about compression and cross entropy stated in the previous section, we demonstrate that near-optimal compression can be constructed from a nearly-accurate autoregressive model of a distribution, and subsequently that a nearly-accurate autoregressive model can be constructed from near-optimal compression. Intuitively, this means that training a model such as an LLM to learn a model that minimizes the loss function $L(\M, t)$ given in Equation \ref{eqn:loss} (i.e., pretraining) is equivalent to training it to learn near-optimal compression via an encoding scheme (such as arithmetic encoding) meeting the requirements of Claim \ref{clm:optimal-encoding}.

\subsection{Accurate Autoregressive Models Imply Near-Optimal Compression}

We begin with the first direction: given an autoregressive model $\M$ for a given distribution such that the expectation of $L(\M, t)$ over samples from the distribution is close to its optimum, one can use a near-optimal encoding scheme (e.g., arithmetic encoding) to construct efficient compression that is similarly near-optimal with respect to coding length.

\begin{lemma}
\label{lemma:llm-compression}
Let $\D$ be a distribution over $\Sigma^m$,  $\M$ be an autoregressive model with induced distribution $D_\M$, and $\Enc$ be a prefix-free encoding scheme such that there exists constant $\alpha$ for which, for any $x\in \sup(\D)$, $\len(\Enc(x)) \leq \log(1/D_\M(x)) + \alpha$. Then, if $\M$ is $\delta$-optimal with respect to $D$, $\Enc$ satisfies:
\[ \Expect_{x \leftarrow D} [\len(\Enc(x))] \leq H(D) + \delta + \alpha \]
\end{lemma}

Proving this lemma is in fact immediate, since we have already demonstrated the necessary groundwork as part of our discussion of compression and KL-divergence in Section \ref{sec:definitions}:

\begin{proof}
By Claim \ref{clm:loss-kl}, we have $KL(D \; || \; D_\M) \leq \delta$ (as the optimum is attained at $D = D_\M$ which results in a divergence of 0). By Claim \ref{clm:optimal-encoding}:
\[ \Expect_{x \leftarrow D}[\len(\Enc(x))] \leq H(D) + KL(D \; || \; D_\M) + \alpha \leq H(D) + \delta + \alpha  \]
\end{proof}

Note that, due to Shannon's lower bound from Corollary \ref{cor:shannon}, we know that this bound on the expected encoding length must be no greater than $\delta + \alpha$ away from the optimal expected encoding length over all prefix-free encodings. 
We can instantiate this using Claim \ref{clm:arithmetic-properties}, which gives $\alpha = 2$, for the following theorem as a direct corollary, noting that Claim \ref{clm:efficient-cumulative} gives the requisite property that the cumulative distribution function is efficiently computable\footnote{While our result does lose a constant number of bits compared to a ``perfect" equivalence, we note that this is in part inevitable, as $\alpha = 1$ is the smallest attainable even by non-monotone compression for an arbitrary family of distributions $\D$.}:

\begin{theorem}
\label{thm:llm-compression-arithmetic}
Given a polynomial-sized alphabet family $\Sigma = \braces{\Sigma_n}_{n \in \Nat}$ with respective ordering $\braces{\prec_n}_{n \in \Nat}$ and polynomial maximum length $m(\cdot)$, there exists an oracle-aided algorithm $\Enc^{(\cdot)}$ so that, for any family of distributions $\D = \braces{D_n}_{n \in \Nat}$, family of autoregressive models $\M = \braces{\M_n}_{n \in \Nat}$ over $\Sigma_n^{m(n)}$ where $\M_n(t_i \; | \; t_{<i})$ is efficiently computable for any $t \in \Sigma_n^{\leq m(n)}$ and induces a distribution $D_{\M_n}$ such that there exists polynomial\footnote{We note that there must exist such a polynomial $\ell(\cdot)$ for any efficiently computable autoregressive model family $\M$, since all next token probabilities must be polynomial-length representable for efficient computability. As noted in the remark following Claim \ref{clm:arithmetic-properties}, one can remove the requirement for such an a priori bound, but at the cost of an exponentially small $\epsilon(n)$ to the encoding length.} $\ell(\cdot)$ for which $- \log(D_{\M_n}(x)) \leq \ell(n)$ for all $x \in \sup(D_{\M_n})$, and function $\delta(\cdot)$ so that, for any $n \in \Nat$, $\M_n$ is $\delta(n)$-optimal with respect to $D_n$, the family of encodings $\Enc = \braces{\Enc^{\M_n, \prec_n}(1^n, \ell(n), \cdot)}_{n \in \Nat}$ is prefix-free, monotone (with respect to $\prec$), and efficient, and additionally, for any $n \in \Nat$:
\[ \Expect_{x \leftarrow D_n} [\len(\Enc_n(x))] \leq H(D_n) + \delta(n) + 2 \]
\end{theorem}

\subsection{Near-Optimal Compression Implies Accurate Autoregressive Models}

Next, we show the other direction: if one has an encoding that compresses a given distribution near-optimally, it is possible to create an autoregressive model that is close to optimal in terms of the expected next-token cross entropy loss $\Expect_{t \leftarrow D} [L(\M, t)]$. Formally:

\begin{theorem}
\label{thm:compression-llm}
Given a polynomial-sized alphabet family $\Sigma = \braces{\Sigma_n}_{n \in \Nat}$ and polynomial maximum length $m(\cdot)$, there exists an oracle-aided algorithm $\M^{(\cdot)}$ such that, for any family of distributions $\D = \braces{D_n}_{n \in \Nat}$, any function $\delta(\cdot)$, and any prefix-free, monotone (with respect to an ordering $\braces{\prec_n}_{n \in \Nat}$) and efficiently encodable family of encodings $\Enc = \braces{\Enc_n}_{n \in \Nat}$ such that, for any $n \in \Nat$,
\[ \Expect_{x \leftarrow D_n} [\len(\Enc_{n}(x))] \leq H(D_n) + \delta(n) \] $\M^{(\cdot)}$ defines a family of autoregressive models $\braces{\M_n}_{n \in \Nat}$ where, for any $n \in \Nat$ and $t \in \Sigma_n^{\leq m(n)}$, $\M_n(t_i \; | \; t_{<i})$ is given by $\M^{\Enc_n, \prec_n}(1^n,  t_i, t_{<i})$, $\M^{(\cdot)}$ is efficiently computable, and the respective model $\M_n$ is $\delta(n)$-optimal with respect to $D_n$.
\end{theorem}

\begin{proof}
(For ease of notation, we omit the parameter $n$ and demonstrate the following for any $n \in \Nat$.)
We wish to define a model $\M^\Enc$ that induces a distribution $D_\M$ with a corresponding efficient random sampler $\Sample : \Bit^{*} \to \Sigma^{m(n)}$ to be defined subsequently (that is, $D_\M = \braces{c \leftarrow \Bit^{*}: \Sample(c)}$). We will start by making the simplifying assumption that 
we know a bound $\ell(\cdot)$ such that $\len(\Enc_{n}(x)) \leq \ell(n)$ for all $x \in \Sigma^{m(n)}$ (note that such a polynomial $\ell$ exists since $\Enc$ is efficient). $\Sample$, given a string $c \in \Bit^{\ell(n)}$, does as follows:
\begin{itemize}
    \item Find the largest $x \in \Sigma^{m(n)}$ with $\Enc(x)\le c$. (findable by binary search over $\Sigma^{m(n)}$).\footnote{If no such $x$ exists, output the smallest element of $\Sigma^{m(n)}$.}
    \item Find the smallest $i$ such that $\Dec(c_{\leq i}) \neq \bot$ (where, leveraging monotonicity of $\Enc$, we define $\Dec$ per Fact \ref{fact:decoding}).
    \item If there is no such $i$, output $x$. Otherwise output $\Dec(c_{\leq i})$.
\end{itemize}
By Fact \ref{fact:decoding} and monotonicity of $\Enc$, $\Dec$ is efficiently computable, and thus $\Sample$ is as well. 
Furthermore, we make the following claim that $\Sample$ is \emph{weakly} monotone (which nonetheless is sufficient for our proof): 

\begin{claim*}
\label{clm:monotone-sampler}
     $\Sample$ is weakly monotone: for any $c_1, c_2 \in \Bit^{\ell(n)}$, if $c_1 < c_2$, then $\Sample(c_1) \preceq \Sample(c_2)$.
\end{claim*}
\begin{proof}
    This follows from the fact that, since $\Enc$ is monotonic and $\Dec$ inverts $\Enc$ on its domain, $\Dec$ must also be monotonic. If both $c_1$ and $c_2$ have prefixes $(c_1)_{\leq i_1}$ and $(c_2)_{\leq i_2}$ with valid decodings, then $(c_1)_{\leq i_1} \leq(c_2)_{\leq i_2}$, and the same is true of the respective decodings by monotonicity. In the case where no prefix of $c_1$ has a valid decoding, $\Sample$ returns the last valid input whose encoding is smaller than $c_1$, which by monotonicity cannot be larger than $\Dec(c_2)$ (or, in the case that $c_2$ also has no prefix with a valid decoding, the last valid input whose encoding is smaller than $c_2$) without contradicting the assumption that $c_1 < c_2$. Finally, if no prefix of $c_2$ has a valid decoding but there is
such a prefix for $c_1$, then
$
\Enc(\Sample(c_1)) \le c_1 < c_2$.
Thus $\Sample(c_1)$ is a candidate in the search for the
largest $x$ satisfying $\Enc(x) \le c_2$. Since
$\Sample(c_2)$ returns that largest candidate, we obtain
$
\Sample(c_1) \preceq \Sample(c_2)$.

\end{proof}

We first need to argue that $\KL{D}{D_\M} \leq \delta$; towards that, we first observe that, for any $x \in \sup(D)$, $D_\M(x) \geq 2^{-\len(\Enc(x))}$. This is true because, whenever $c$ is drawn such that it has $\Enc(x)$ as a prefix (which happens with precisely probability $2^{-\len(\Enc(x))}$), that prefix is shorter than the maximum length $\ell(n)$ by assumption, and it is also the shortest correctly decodable prefix by prefix-freeness of $\Enc$, and thus $x$ will be sampled. From here, Fact \ref{fact:cross-kl} implies:
\[ \KL{D}{D_\M} = H(D, D_\M) - H(D) \]
But for any $x \in \sup(D)$, since $D_\M(x) \geq 2^{-\len(\Enc(x))}$:
\[ -\log D_\M(x) \leq \len(\Enc(x)) \]
and thus, by our original assumption about the compression loss of $\Enc$:
\[ - \Expect_{x \leftarrow D}[\log D_\M(x)] \leq \Expect_{x \leftarrow D}[\len(\Enc(x))] \leq H(D) + \delta \]
But $- \Expect_{x \leftarrow D}[\log D_\M(x)]$ is exactly the cross-entropy $H(D, D_\M)$, so:
\[ \KL{D}{D_\M} \leq (H(D) + \delta) - H(D) = \delta \]
as desired.

Thus far, we have given a description of the induced distribution $D_\M$ in terms of the random sampler $\Sample$, but not constructed the actual model $\M^\Enc$ nor shown that it is efficiently computable. Fortunately, Claim \ref{clm:efficient-cumulative} provides both of these as long as we can efficiently compute $\Pr_{x \leftarrow D_\M} [x \preceq y]$ for any $y \in \Sigma^{m(n)}$.\footnote{While Claim \ref{clm:efficient-cumulative} technically only provides efficient computability on prefixes with non-zero probability of occurrence, we note that we can trivially assign an arbitrary next-token distribution to prefixes with zero probability without changing $D_\M$.} However, we claim that, using the efficiency of $\Sample$ and Claim \ref{clm:monotone-sampler}, we can do precisely this:

\begin{claim*}
\label{clm:cdf}
    For any $y \in \Sigma^{m(n)}$, $\Pr_{x \leftarrow D_\M} [x \preceq y]$ can be efficiently computed.
\end{claim*}

\begin{proof}
    If we can efficiently determine the largest $c \in \Bit^{\ell(n)}$ such that $\Sample(c) \leq y$ (taking $c = 0^{\ell(n)}$ if there is no such $c$), then we can immediately determine that (taking $c$ as an integer) $\Pr_{x \leftarrow D_\M} [x \preceq y] = (c+1)/2^{\ell(n)}$, as sampling from $D_\M$ is by definition equivalent to sampling $\Sample(c)$ for uniformly random $c \in \Bit^{\ell(n)}$. And weak monotonicity from Claim \ref{clm:monotone-sampler} is sufficient to show that we can find such a $c$ using binary search over $\Bit^{\ell(n)}$, which must run in time polynomial in $n$ since it requires at most $\ell(n)$ iterations (which we know is polynomial) and each iteration requires computing $\Sample$ and comparing elements of $\Sigma^{m(n)}$, both of which are efficient.
\end{proof}
Thus, by Claim \ref{clm:efficient-cumulative}, we have constructed an efficiently computable autoregressive model $\M$ with induced distribution $D_\M$ that, as we have argued, suffices to complete the argument given an a priori bound of $\ell$ on the length of $\Enc(x)$.

\paragraph{Removing $\ell$.} Next, we explain how to remove the assumption that $\ell$ is known. First, we modify $\Sample$ to, rather than sampling $c \leftarrow \Bit^{\ell(n)}$, instead lazily sample $c \leftarrow \Bit^*$: specifically, during the binary search, each time $\Sample$ compares $c$ to some output $\Enc(y)$ for $y \in \Sigma^m$, if the length of the sampled $c$ is shorter than $\len(\Enc(y))$, $\Sample$ draws new uniformly random bits of $c$ up to length $|\Enc(y)|$. We argue that the output of $\Sample$ is identical. Letting $x$ be the largest $x \in \Sigma^{m}$ such that $\Enc(x) \le c$, we notice that number of bits of $c$ sampled by $\Sample$ is at least $|\Enc(x)|$ by construction. Moreover, if (the $\ell$-bit long)  $c$ has a prefix $c_{\le i}$ which is a valid codeword, then by monotonicity $c_{\le i}$ is a prefix of $\Enc(x)$ (otherwise, $\Enc(x) < c_{\leq i} \leq c$, contradicting $x$ being the largest $x \in \Sigma^{m}$ such that $\Enc(x) \le c$). Thus, in this case, $\Sample$ must return $\Dec(c_{\leq i})$ as before; if there is no such prefix, then $\Sample$ by construction finds and returns $x$, as desired.

Similarly, when computing the CDF in Claim \ref{clm:cdf}, finding the largest $c$ such that $\Sample(c)\le y$ is done by binary search where the length of the string $c$ determined by the number of random bits $\Sample$ reads during its execution on $c$. In the $i^\text{th}$ step of the binary search, we know the prefix $c_{<i}$, and we learn the $i^\text{th}$ bit by comparing the output of $\Sample(c_{<i}||10^*)$ to $y$ (if $y$ is smaller, the next bit is $0$, where if $y$ is larger or equal the next bit is $1$). We terminate the search when, on the current prefix of $c$, the execution of $\Sample$ halts on the prefix without drawing any additional random bits. After $|c|$ steps, we learn the entire string $c$, and can then compute the CDF $\Pr_{x \leftarrow D_\M} [x \preceq y]$ as before by setting $\ell=|c|$.
\end{proof}

As a final note, the random sampler $\Sample$ we constructed gives the following corollary, which intuitively states that, given $\delta$-optimal compression $\Enc$ for a distribution, it can be efficiently and $\delta$-approximately sampled:

\begin{corollary}
\label{cor:compression-llm}
Given a polynomial-sized alphabet family $\Sigma = \braces{\Sigma_n}_{n \in \Nat}$ and polynomial maximum length $m(\cdot)$, there exists an oracle-aided sampling algorithm $\Sample^{(\cdot)}$ such that, for any family of distributions $\D = \braces{D_n}_{n \in \Nat}$, function $\delta(\cdot)$, and prefix-free, monotone, and efficiently encodable family of encodings $\Enc = \braces{\Enc_n}_{n \in \Nat}$ where, for any $n \in \Nat$,
\[ \Expect_{x \leftarrow D_n} [\len(\Enc_{n}(x))] \leq H(D_n) + \delta(n) \]
$\Sample^{(\cdot)}$ defines an efficiently computable family of randomized algorithms $\braces{\Sample_n}_{n \in \Nat}$ such that there exists a family of distributions $\braces{D^*_n}_{n \in \Nat}$ where, for any $n \in \Nat$, $\Sample_n$ samples from $D^*_n$ and $KL(D_n \; || \; D^*_n) \leq \delta(n)$.
\end{corollary}

\section{One-Way Functions and the Monotonicity Requirement}

Given that our results in the previous section show equivalence between next-token prediction and specifically \emph{monotone} compression, the question arises of whether one can show a similar equivalence for non-monotone compression, or more specifically whether there are certain conditions under which one can (or cannot) show such an equivalence. In this section, we demonstrate such conditions based on a standard cryptographic assumption: the existence of one-way functions. Whereas Theorem \ref{thm:compression-llm} shows that monotone compression can be used to construct an autoregressive model with the same loss, we show in Section \ref{sec:owf-existence} that, given the existence of one-way functions, there exists a family of distributions with non-monotone compression where no efficiently computable autoregressive model can even achieve a cross-entropy loss within an arbitrary constant of the compression loss; the transformation necessarily incurs a loss linear in the input length $n$. Thus, if one-way functions exist, there is a distinct separation between the results that can be achieved with monotone compression compared to non-monotone compression.

Conversely, we show in Section \ref{sec:owf-nonexistence} that, if one-way functions do not exist, classic results on next-bit prediction \cite{il89} can be leveraged to show an analogue of Theorem \ref{thm:compression-llm} without the requirement for monotonicity. At a high level, if one-way functions do not exist, \cite{il89} shows that any (binary) distribution must admit an approximately accurate next-bit predictor. The transformation is not a perfect equivalence, as it creates a small amount of additional \emph{statistical distance} loss in contrast to Theorem \ref{thm:compression-llm} (where there is only cross-entropy loss equal to the compression loss); nonetheless, it still constitutes a significant difference compared to what is attainable if one-way functions exist.

\subsection{Preliminaries}

In order to formally state our results, we first precisely define one-way functions.

\begin{definition}
A family of functions $\F = \braces{f_n}_{n \in \Nat}$, where $f_n : \Bit^n \rightarrow \Bit^n$, is \textbf{one-way} if $\F$ is efficiently computable but, for any non-uniform polynomial-time inverter $\A$, there is a negligible function $\epsilon(\cdot)$ such that for any $n \in \Nat$:
\[ \pr{x \leftarrow \Bit^n : f(\A(1^n, f(x))) = f(x)} \leq \epsilon(n) \]

If the above instead holds for only \emph{infinitely many} $n \in \Nat$ for any non-uniform polynomial-time inverter $\A$, we say that $\F$ is \textbf{infinitely-often one-way}.
\end{definition}

That is, one-way functions are efficiently computable, but ``hard to invert" in the sense that no computationally efficient algorithm, given the image of a uniformly selected input, can find any inverse of that image with probability significantly better than brute force. The existence of one-way functions is a long-standing open conjecture as it would imply $\mathsf{P} \neq \mathsf{NP}$, but they are nonetheless widely used as a fundamental building block of modern cryptography. The existence of one-way functions is also well-known to imply the existence of \emph{pseudorandom generators}:

\begin{definition}
For some output length $\ell(n) > n$, a family of efficiently computable functions $\G = \braces{g_n}_{n \in \Nat}$ is referred to as a family of \emph{pseudorandom generators} (PRGs) if, for any polynomial-time distinguisher $D$ with output in $\Bit$, there exists a negligible function $\epsilon(\cdot)$ such that for any $n \in \Nat$:
\[ \left| \pr{x \leftarrow U_n : D(1^n, g_n(x)) = 1} - \pr{y \leftarrow U_{\ell(n)} : D(1^n, y) = 1} \right| \leq \epsilon(n) \]
where $U_k$ represents the uniform distribution over $\Bit^k$.
\end{definition}

That is, the output of a PRG on a uniform input should be computationally indistinguishable from uniform randomness.

\subsection{One-Way Functions Imply Separation}
\label{sec:owf-existence}

We can, in fact, use PRGs to show a separation between monotone and non-monotone compression: in contrast to Theorem \ref{thm:compression-llm}, which shows that any \emph{monotone} compression scheme for a distribution can be transformed into a next-bit predictor for that distribution whose negative log-likelihood is identical to the compression loss, we can show that, assuming the existence of infinitely-often one-way functions (and thus PRGs), there is a family of distributions $\D$ and a \emph{non-monotone} compression scheme that is \emph{optimal} for $\D$ but is such that no efficiently computable autoregressive model exists that can predict $\D$ with even an arbitrary constant negative log-likelihood:

\begin{theorem}
\label{thm:owf-separation}
Assuming the existence of one-way functions (respectively, infinitely-often one-way functions), there exists a family of efficiently samplable distributions $\D = \braces{D_n}_{n \in \Nat}$ and a prefix-free family of encodings $\Enc = \braces{\Enc_n}_{n \in \Nat}$ with efficient encoding and decoding such that $\Expect_{x \leftarrow D_n} [\len(\Enc_{n}(x))] = H(D_n)$ for any $n \in \Nat$, but, given any efficiently computable family of autoregressive models $\braces{\M_n}_{n \in \Nat}$, for any constant $\epsilon > 0$ and all sufficiently large (resp., infinitely many) $n \in \Nat$, the respective output distribution $D_{\M_n}$ is such that $KL(D_n \; || \; D_{\M_n}) \geq (1 - \epsilon)n$.\footnote{This implies that, for any constant $\delta$, $\M$ cannot be $\delta$-optimal with respect to $\D$.}
\end{theorem}
Below we prove the theorem assuming the existence of one-way functions. We note that the proof holds nearly identically for the case of infinitely-often OWFs, with the exception that it applies only to infinitely many $n \in \Nat$ rather than all sufficiently large $n$; we note the distinctions during the proof.
\begin{proof}
As the existence of one-way functions implies the existence of PRGs (respectively, infinitely-often OWFs imply \emph{infinitely-often PRGs}, or PRGs where the condition holds for only infinitely many $n \in \Nat$), let $\G = \braces{g_n}_{n \in \Nat}$ be a PRG with output length $\ell(n) = 2n$. Take our alphabet of tokens to be $\Sigma = \Bit$. Let $\D$ be given by $D_n = \braces{x \leftarrow U_n : (g_n(x), x)}$. Let $\Enc$ be given by $\Enc_n(g_n(x), x) = x$.\footnote{That is, simply take the last $n$ bits of the input. This can be efficiently decoded as well by simply computing $g_n(x)$; one can extend it to the full domain of arbitrary $3n$-bit inputs by letting $\Enc_n(y, x) = y || x$ if $y \neq g_n(x)$.} Clearly, $\Enc$ satisfies prefix-freeness (since every codeword has the same length), efficient encoding and decoding, and has $\Expect_{x \leftarrow D_n} [\len(\Enc_{n}(x))] = H(D_n) = n$.

Now fix $\epsilon > 0$ and assume for the sake of contradiction that there exists a family of autoregressive models $\braces{\M_n}_{n \in \Nat}$ such that $KL(D_n \; || \; D_{\M_n}) < (1 - \epsilon)n$ for infinitely many (resp., all sufficiently large) $n \in \Nat$. We use this to construct a distinguisher $\A$ 
as follows: Take constant $c > \log (2/\epsilon)$, and, given an input $y$ (which will either be drawn as $U_{2n}$ or as $g_n(U_n)$), define $D_{\M_n}(y, *)$
as the probability for the first $2n$ bits of an output from $\M_n$ to be equal to $y$, and let $\A$ output $1$ when $D_{\M_n}(y, *) \geq 2^{c-2n}$ and $0$ otherwise. We make the following claims to complete the proof by contradiction:

\begin{claim*}
$\pr{y \leftarrow U_{2n} : \A(y) = 1} \leq 2^{-c}$.
\end{claim*}

\begin{proof}
$\Expect_{y \leftarrow U_{2n}}[D_{\M_n}(y, *)] = 2^{-2n}$, so this follows from Markov's inequality.
\end{proof}

\begin{claim*}
For any constant $\epsilon' > 0$, $\pr{x \leftarrow U_{n} : \A(g_n(x)) = 1} > \epsilon/2 - \epsilon'$ for all sufficiently large $n \in \Nat$ where $KL(D_n \; || \; D_{\M_n}) < (1 - \epsilon)n$.
\end{claim*}

\begin{proof}
This follows from the assumption that $KL(D_n \; || \; D_{\M_n}) < (1 - \epsilon)n$, and thus, by data processing, that $KL((D_n)_{\leq 2n} \; || \; (D_{\M_n})_{\leq 2n}) < (1 - \epsilon)n$ (where $(D_n)_{\leq 2n}$ represents the distribution of the first $2n$ bits of $D_n$, and respectively for $D_{\M_n}$). Thus:
\[ KL((D_n)_{\leq 2n} \; || \; (D_{\M_n})_{\leq 2n}) = \Expect_{y \leftarrow D_n}[\log(1/D_{\M_n}(y, *))] - H((D_n)_{\leq 2n}) < (1 - \epsilon)n \]
and so
\[ \Expect_{x \leftarrow U_n}[\log(1/D_{\M_n}(g_n(x), *))] = \Expect_{y \leftarrow D_n}[\log(1/D_{\M_n}(y, *))] < (1 - \epsilon)n + H((D_n)_{\leq 2n}) \leq (2-\epsilon)n \]
since $H((D_n)_{\leq 2n}) \le H(D_n) = n$. Now, by construction of $\A$, $\A$ outputs $1$ whenever $\log(1/D_{\M_n}(y, *)) < 2n-c$; by Markov's inequality and the above bound on the expectation, the probability that this is \emph{not} true (i.e., that $\A$ outputs $0$) when $y$ is given by $g_n(x)$ for $x \leftarrow U_n$ can be at most $\frac{(2-\epsilon)n}{2n-c}$, so the probability that $\A$ outputs $1$, then, must be at least 
\[ 1 - \frac{(2-\epsilon)n}{2n-c} = \frac{\epsilon n - c}{2n - c}\]
which asymptotically approaches $\epsilon/2$ as $n \rightarrow \infty$, implying the conclusion.
\end{proof}

Given these, since we assumed by construction that $2^{-c}$ is a constant smaller than $\epsilon /2$, there is a non-negligible (constant) separation between the two probabilities for infinitely many (resp., any sufficiently large) $n$, which contradicts the assumption that $\G$ is a PRG (resp., infinitely-often PRG) and completes the proof. 
\end{proof}

\subsection{Nonexistence of One-Way Functions Implies (Approximate) Equivalence}
\label{sec:owf-nonexistence}

In contrast to the above, we show that, if infinitely-often one-way functions do not exist, then given any compression scheme for a distribution over a binary support, one can construct an autoregressive model (i.e., a next-bit predictor) for that distribution which is ``approximately" as good. This follows from a theorem proven in \cite{il89}, and in fact we observe that it does not require full non-existence of one-way functions, only that a \emph{specific} function relative to the given distribution is easy to invert. First, we must define what we mean by a model being ``approximately as good" in this case---unlike our previous results, this will involve a loss of statistical distance in addition to the inherent loss of negative log-likelihood resulting from the underlying compression scheme's loss:

\begin{definition}
\label{def:next-bit-predictor}
  For alphabet $\Sigma$ and distribution $D$ over $\Sigma^{m}$, we say that an autoregressive model $\M$ is \textbf{$\epsilon$-approximately $\delta$-optimal} with respect to $D$ if there exists a model $\M'$ that is $\delta$-optimal with respect to $D$ and, $\Delta(\D_\M,\D_{\M'})\le \epsilon$. 
\end{definition}

We can then leverage the following theorem from \cite{il89}:

\begin{theorem}[Infinitely-often one-way functions and next-bit predictors \cite{il89}]\label{thm:owf-inverter-to-predictor} Let $\D = \braces{D_n}_{n \in \Nat}$ be a family of efficiently samplable distributions where each $D_n$ distributed over $\Bit^{m(n)}$ for some polynomial $m(n)$. Then, for any function $\epsilon(n) \in 1/\poly(n)$, there exists a function $f$ such that, if $f$ is not an infinitely-often one-way function, then there exists an $\epsilon(n)$-approximate ($0$-optimal) efficiently computable autoregressive model for $\D$.
\end{theorem}

\begin{proof}
\cite{il89} showed that for every family of efficiently samplable distributions $\D$ and for any $\epsilon'(n) \in 1/\poly(n)$, there exists a function $f$ such that, if $f$ is not an infinitely-often one-way function, then there exists a (non-uniform) efficient \emph{randomized} next bit-predictor $\Pred$ with \[\Pr_{t\leftarrow D_{n}}[\Pred(t_i \; | \; t_{<i})  \in D_n(t_i \; | \; t_{<i})(1-\epsilon', 1+\epsilon')]\ge 1-\epsilon'\]
for every sufficiently large $n \in \Nat$. We first notice that $\Pred$ can be made deterministic by fixing the randomness to the best possible string for each index $i$ as part of the non-uniform advice.

Next, we notice that $\Pred$ is already an $\epsilon=O(\epsilon'\cdot m(n))$-approximate model for $\D$, as the above promise on $\Pred$ implies that its output is close in statistical distance to the distribution $D_n$.
\end{proof}

\begin{remark}
In fact, Theorem \ref{thm:owf-inverter-to-predictor} can be made stronger: one can show the existence of such a model based on an inverter to a \emph{specific} function (as opposed to complete non-existence of infinitely-often one-way functions, which would imply the existence of an inverter for \emph{every} function). In particular, \cite{il89} showed that for any $\epsilon(n) \in 1/\poly(n)$, there exists a pair of efficient oracle-aided algorithms $(f^{(\cdot)},\Pred^{(\cdot)})$ such that, given a sampler $\Sample$ for an ensemble of samplable distributions $\D$, if there exists a polynomial-time algorithm $\A$ that, for all sufficiently large $n \in \Nat$, inverts $f^{\Sample}$ with probability $1/2$, then $\Pred^A$ is an $\epsilon$-approximate next-bit predictor for $\D$. (In particular, $f^{\Sample}(x_1,\dots,x_t)= g^{\Sample}(x_1),\dots,g^{\Sample}(x_t)$, where $g^{\Sample}(r,h,i,j)=(i,j,h,\Sample(r)_{\le i}, h(r)_{\le j})$, $h$ is a 2-universal hash function, and $t \in O(n/\epsilon)$.)
\end{remark}

With the above theorem, given a non-monotone compression scheme $\Enc$ which is $\delta$-optimal for some distribution family $\D$, we can leverage part of the construction from Theorem \ref{thm:compression-llm} even though $\Enc$ is not monotone. We prove the following theorem.

\begin{theorem}
\label{thm:owf-inverter-to-predictor-2}
If infinitely-often one-way functions do not exist, then the following is true for any $\epsilon(n) \in 1/\poly(n)$:

For every  family $\D = \braces{D_n}_{n \in \Nat}$ distributions  and a compression scheme $(\Enc,\Dec)$ with efficient encoding and decoding that is $\delta$-optimal for $\D$, there exists an efficiently computable autoregressive models $\braces{\M_n}_{n \in \Nat}$ where, for any $n \in \Nat$ the respective model $\M_n$ is $\epsilon(n)$-approximately $\delta(n)$-optimal with respect to $D_n$.
\end{theorem}
\begin{proof}
We first notice that, for any $n$ we can construct a distribution $\D'_n$ such that $KL(D_n,D'_n)\le \delta(n)$.
Specifically, we showed in the proof of Theorem \ref{thm:compression-llm} that the distribution $D_\M$, which can be thought of as sampling a uniformly random string and decoding the (unique) codeword that is a prefix of that string, is not only efficiently samplable even when $\Enc$ is non-monotone (it requires only that $\Enc$ has an efficient decoding algorithm $\Dec$), but induces a model $\M^*$ which is $\delta$-optimal with respect to $\D$. (We note that $\Sample$ as written requires binary search and thus monotonicity to be efficient, but, if weak monotonicity is not needed, one can replace this step while preserving $\delta$-optimality by returning an arbitrary output if the given random string has no validly decodable prefix.)

Next, while $\M^*$ is not necessarily efficiently computable if $\Enc$ (and thus $\Sample$) is non-monotone, Theorem \ref{thm:owf-inverter-to-predictor} shows directly that, if infinitely-often one-way functions do not exist, then for any function $\epsilon(n) \in 1/\poly(n)$ there is a model $\M$ that is $\epsilon(n)$-approximate with respect to $D_\M$ and thus an $\epsilon(n)$-approximate $\delta$-optimal model for $\D$.

\end{proof}

\section{Next-Bit Pseudoentropy and Monotone Incompressibility}

In this section, we discuss an interesting application of our results to cryptography---specifically, the relation between the compressibility of a distribution and its \emph{pseudoentropy} \cite{hill}, which can be thought of as a computational analogue of entropy. Informally, a distribution $D$ has pseudoentropy at least $h$ if it is computationally indistinguishable from some distribution $D'$ with Shannon entropy $H(D') \geq h$. Restricting to the case of binary distributions ($\Sigma = \Bit$), we will formally consider a related notion of \emph{next-bit} pseudoentropy, or pseudoentropy with respect to next-bit prediction:

\begin{definition}[\cite{vadhan}]
\label{def:pseudoentropy}
Given jointly distributed families of random variables $\X = \braces{X_n}_{n \in \Nat}$ and $\B = \braces{B_n}_{n \in \Nat}$, where $\B$ takes values in $\Bit$, we say that $\B$ has (non-uniform) \textbf{conditional pseudoentropy} at least $h(\cdot)$ given $X$ if, for any polynomial $p(\cdot)$, there exists a $\Bit$-valued family $\C = \braces{C_n}_{n \in \Nat}$ of random variables jointly distributed with $(\X, \B)$ such that:
\begin{enumerate}
    \item For all sufficiently large $n \in \Nat$, $H(C_n | X_n) \geq h(n) - 1/p(n)$
    \item The families $(\X, \B)$ and $(\X, \C)$ are computationally indistinguishable.\footnote{That is, for any efficient non-uniform distinguisher $D$, there exists negligible $\epsilon(\cdot)$ such that, for all sufficiently large $n \in \Nat$, $\pr{(x, b) \leftarrow (X_n, B_n): D(1^n, x, b) = 1}$ and $\pr{(x, c) \leftarrow (X_n, C_n): D(1^n, x, c) = 1}$ differ by no more than $\epsilon(n)$.}
\end{enumerate}

Given a polynomial $m(\cdot)$ and family of random variables $\B = \braces{B_n}_{n \in \Nat}$ where $B_n$ takes values in $\Bit^{m(n)}$, we say that $\B$ has (non-uniform) \textbf{next-bit pseudoentropy} at least $h(\cdot)$ if, letting family $\I = \braces{I_n}_{n \in \Nat}$ be such that $I_n$ is uniformly distributed over $[m(n)]$, then the conditional pseudoentropy of the family $\braces{(B_n)_{I_n}}_{n \in \Nat}$ given $\braces{(B_n)_{<I_n}}_{n \in \Nat}$ is at least $h(n)/m(n)$.
\end{definition}

\cite{haitner} shows a partial characterization of the relation between next-bit pseudoentropy and compressibility of a family of distributions (a notion originally proposed by \cite{yao}). Specifically, they show that incompressibility implies pseudoentropy (\cite{haitner}, Lemma 36): if a family of distributions $\B = \braces{B_n}_{n \in \Nat}$ over $\Bit^{m(n)}$ cannot be compressed with some length $k(n)$, then it must have next-bit pseudoentropy at least $k(n) - 2$. However, they also show that the converse is \emph{not} true; in fact, there exists a separation (\cite{haitner}, Prop. 39, Lemma 42) showing that, assuming the existence of one-way functions, next-bit pseudoentropy cannot imply incompressibility unless it is strengthened to a notion of ``strong" next-bit pseudoentropy which requires next-bit pseudoentropy of $\braces{\Enc(B_n)}_{n \in \Nat}$ for any efficient encoding $\Enc(\cdot)$.

Our results, in fact, complete the characterization of \emph{standard} next-bit pseudoentropy by showing that, up to a constant, it is equivalent to \emph{monotone} incompressibility. This is illustrated by the following corollary of Corollary 3.9 of \cite{vadhan}:

\begin{corollary}[\cite{vadhan, haitner}]
\label{cor:pseudoentropy}
Given polynomial maximum length $m(\cdot)$ and family of distributions $\D = \braces{D_n}_{n \in \Nat}$ where $D_n$ is a distribution over $\Bit^{m(n)}$, for some function $\delta(\cdot)$, $\D$ has next-bit pseudoentropy at least $H(D_n) + \delta(n)$ if and only if, for any polynomial $p(\cdot)$, for every efficiently computable autoregressive model family $\M = \braces{M_n}_{n \in \Nat}$ and every sufficiently large $n \in \Nat$, $M_n$ is not  $(\delta(n)-1/p(n))$-optimal with respect to $D_n$.
\end{corollary}

The original corollary in \cite{vadhan} shows an equivalence between $\B$ given $\X$ (as given in Definition \ref{def:pseudoentropy}) having conditional pseudoentropy $H(B_n | X_n) + \delta(n)$ and $\B$ being \emph{$\delta(n)$-KL-hard} given $\X$---that is, that there exists no efficient predictor $P(x)$ that attempts to predict $B$ given $X$ and has $KL((X, B) \; || \; (X, P(X))) \leq \delta(n)$. \cite{haitner} observed that, given a family $\D$ of distributions, one can apply this to a uniformly random bit in a distribution $D_n$ (as in the definition of next-bit entropy) and use the additivity of entropy and KL-divergence chain rule to immediately imply that $\D$ not having ``non-trivial" next-bit entropy $H(D_n) + \delta(n)$ implies a next-bit predictor $\M$ for \emph{any} bit of the distribution $D_n$ (i.e., an autoregressive model, using our terminology) for which $KL(D_n \; || \; D_{\M_n}) \leq \delta(n)$, or, in other words, $\M$ is $\delta(n)$-optimal with respect to $\D$.

Returning to the equivalence between next-bit pseudoentropy and monotone compression and applying Corollary \ref{cor:pseudoentropy}, we can now observe that this equivalence is intuitively the same as equivalence between next-bit \emph{prediction} and monotone compression, which, in fact, is precisely what we have shown in Section \ref{sec:equivalence}. The first direction is immediate from (the contrapositive of) Theorem \ref{thm:llm-compression-arithmetic} and Corollary \ref{cor:pseudoentropy}, though we note that it is also already proven implicitly in Lemma 36 of \cite{haitner} because their construction of a compression scheme is monotone.

\begin{corollary}
\label{cor:incompressibility-pseudoentropy}
Given polynomial maximum length $m(\cdot)$ and family of distributions $\D = \braces{D_n}_{n \in \Nat}$ where $D_n$ is a distribution over $\Bit^{m(n)}$, if, for some function $k(\cdot)$, $\D$ cannot be monotonically compressed to length $k(n)$ infinitely often (namely, for any efficient, monotone, prefix-free encoding $\Enc$, $E_{x\gets \D_n}[len(\Enc(x))]\ge k(n)$ for all but finitely many $n \in \Nat$),
then $\D$ has next-bit pseudoentropy at least $k(n)-2$.
\end{corollary}

Specifically, if $\D$ has insufficient pseudoentropy, then by Corollary \ref{cor:pseudoentropy} it must have a $(k(n) - 2 - H(D_n)-1/\poly(n))$-optimal efficiently computable autoregressive model family $\M$, and Theorem \ref{thm:llm-compression-arithmetic} demonstrates that arithmetic encoding gives efficient and length-$(k(n)-1/\poly(n))$ compression. The other direction is also immediate from our Theorem \ref{thm:compression-llm} and Corollary \ref{cor:pseudoentropy}:

\begin{corollary}
\label{cor:pseudoentropy-incompressibility}
Given polynomial maximum length $m(\cdot)$ and family of distributions $\D = \braces{D_n}_{n \in \Nat}$ where $D_n$ is a distribution over $\Bit^{m(n)}$, if, for some function $k(\cdot)$, $\D$ has next-bit pseudoentropy at least $k(n)$, then, for any polynomial $p(\cdot)$, $\D$ cannot be monotonically compressed to length $k(n) - 1/p(n)$.
\end{corollary}

Since Theorem \ref{thm:compression-llm} shows that $\D$ having length-$(k(n)-1/\poly(n))$ compression implies a $(k(n)-H(D_n)-1/\poly(n))$-optimal autoregressive model family $\M$, the other direction of Corollary \ref{cor:pseudoentropy} implies that $\D$ cannot have next-bit pseudoentropy at least $k(n)$ in this case.

\section*{Generative AI Use Disclosure}

While finalizing this paper for submission, we used generative AI tools to review our draft in order to fix typographical errors, evaluate technical correctness of our proofs and suggest fixes where applicable, and suggest improvements to the structure and readability of our introduction. We emphasize that all results and writing presented in this paper were produced by the authors, except parts of the technical overview (Section \ref{sec:overview}), which were rewritten using AI tools and subsequently reviewed and edited by the authors.

\bibliographystyle{alpha}

\bibliography{watermarking}

@article{shannon,
  author={Shannon, Claude E.},
  journal={{The Bell System Technical Journal}}, 
  title={{A Mathematical Theory of Communication}}, 
  year={1948},
  volume={27},
  number={3},
  pages={379-423},
  doi={10.1002/j.1538-7305.1948.tb01338.x}}

@article{TrevisanVZ05,
	title={{Compression of Samplable Sources}},
	author={Trevisan, Luca and Vadhan, Salil and Zuckerman, David},
	journal=Jcomplexity,
	volume={14},
	number={3},
	pages={186--227},
	year={2005},
}

@inproceedings{gs85,
author = {Goldberg, Andrew V. and Sipser, Michael},
title = {{Compression and Ranking}},
year = {1985},
isbn = {0897911512},
publisher = {Association for Computing Machinery},
address = {New York, NY, USA},
url = {https://doi.org/10.1145/22145.22194},
doi = {10.1145/22145.22194},
booktitle = {Proceedings of the Seventeenth Annual ACM Symposium on Theory of Computing},
pages = {440–448},
numpages = {9},
location = {Providence, Rhode Island, USA},
series = {STOC '85}
}

@ARTICLE{arithmetic,
  author={Gilbert, E. N. and Moore, E. F.},
  journal={The Bell System Technical Journal}, 
  title={{Variable-Length Binary Encodings}}, 
  year={1959},
  volume={38},
  number={4},
  pages={933-967},
  doi={10.1002/j.1538-7305.1959.tb01583.x}}

@article{turing,
    author = {Turing, A. M.},
    title = {{Computing Machinery and Intelligence}},
    journal = {Mind},
    volume = {LIX},
    number = {236},
    pages = {433-460},
    year = {1950},
    month = {10},
    issn = {0026-4423},
    doi = {10.1093/mind/LIX.236.433},
    url = {https://doi.org/10.1093/mind/LIX.236.433},
    eprint = {https://academic.oup.com/mind/article-pdf/LIX/236/433/61209000/mind_lix_236_433.pdf},
}

@ARTICLE{shannon2,
  author={Shannon, C. E.},
  journal={The Bell System Technical Journal}, 
  title={{Prediction and entropy of printed English}}, 
  year={1951},
  volume={30},
  number={1},
  pages={50-64},
  doi={10.1002/j.1538-7305.1951.tb01366.x}}

@inproceedings{mahoney,
    author = {Mahoney, Matthew V.},
    title = {{Text compression as a Test for Artificial Intelligence}},
    year = {1999},
    isbn = {0262511061},
    publisher = {American Association for Artificial Intelligence},
    address = {USA},
    booktitle = {Proceedings of the Sixteenth National Conference on Artificial Intelligence and the Eleventh Innovative Applications of Artificial Intelligence Conference Innovative Applications of Artificial Intelligence},
    pages = {970},
    location = {Orlando, Florida, USA},
    series = {AAAI '99/IAAI '99}
}

@inproceedings{hernandez,
author = {Hernandez-Orallo, Jose},
year = {1998},
month = {07},
pages = {146–163},
title = {{A Formal Definition of Intelligence Based on an Intensional Variant of Algorithmic Complexity}},
booktitle = {Proceedings of International Symposium of Engineering of Intelligent Systems (EIS'98)},
}

@inproceedings{deletang,
title={Language Modeling Is Compression},
author={Gr\'egoire Del\'etang and Anian Ruoss and Paul-Ambroise Duquenne and Elliot Catt and Tim Genewein and Christopher Mattern and Jordi Grau-Moya and Li Kevin Wenliang and Matthew Aitchison and Laurent Orseau and Marcus Hutter and Joel Veness},
booktitle={The Twelfth International Conference on Learning Representations},
year={2024},
url={https://openreview.net/forum?id=jznbgiynus}
}

@inproceedings{huang,
title={Compression Represents Intelligence Linearly},
author={Yuzhen Huang and Jinghan Zhang and Zifei Shan and Junxian He},
booktitle={First Conference on Language Modeling},
year={2024},
url={https://openreview.net/forum?id=SHMj84U5SH}
}

@inproceedings{zhang,
author = {Zhang, Yeqin and Zhao, Yizheng and Hu, Chen and Jiao, Binxing and Jiang, Daxin and Miao, Ruihang and Nguyen, Cam Tu},
year = {2026},
pages = {},
booktitle={Proceedings of the Fortieth Annual AAAI Conference on Artificial Intelligence},
title = {{Learning to Compress: Unlocking the Potential of Large Language Models for Text Representation}},
doi = {10.48550/arXiv.2511.17129}
}

@article{tsai,
author = {Tsai, Chen-Han},
year = {2026},
month = {01},
pages = {},
journal={arXiv preprint}, 
title = {{Revisiting Data Compression with Language Modeling}},
doi = {10.48550/arXiv.2601.02875}
}

@INPROCEEDINGS{frank,
  author={Frank, Eibe and Chui, Chang and Witten, Ian H.},
  booktitle={Proceedings DCC 2000. Data Compression Conference}, 
  title={{Text categorization using compression models}}, 
  year={2000},
  volume={},
  number={},
  pages={555-},
  doi={10.1109/DCC.2000.838202}}

@inproceedings{schmidhuber,
    author = {Schmidhuber, J{\"u}rgen and Heil, Stefan},
    title = {{Predictive Coding with Neural Nets: Application to Text Compression}},
    year = {1994},
    publisher = {MIT Press},
    address = {Cambridge, MA, USA},
    booktitle = {Proceedings of the 8th International Conference on Neural Information Processing Systems},
    pages = {1047–1054},
    numpages = {8},
    location = {Denver, Colorado},
    series = {NIPS'94}
    }

@ARTICLE{schmidhuber2,
  author={Schmidhuber, J{\"u}rgen and Heil, Stefan},
  journal={IEEE Transactions on Neural Networks}, 
  title={{Sequential Neural Text Compression}}, 
  year={1996},
  volume={7},
  number={1},
  pages={142-146},
  doi={10.1109/72.478398}}

@inproceedings{mahoney2,
author = {Mahoney, Matthew V.},
title = {{Fast Text Compression with Neural Networks}},
year = {2000},
isbn = {1577351134},
publisher = {AAAI Press},
booktitle = {Proceedings of the Thirteenth International Florida Artificial Intelligence Research Society Conference},
pages = {230–234},
numpages = {5}
}

@INPROCEEDINGS{schiopu,
  author={Schiopu, Ionut and Liu, Yu and Munteanu, Adrian},
  booktitle={2018 Picture Coding Symposium (PCS)}, 
  title={{CNN-based Prediction for Lossless Coding of Photographic Images}}, 
  year={2018},
  volume={},
  number={},
  pages={16-20},
  doi={10.1109/PCS.2018.8456311}}

@inproceedings{vadhan,
author = {Vadhan, Salil and Zheng, Colin Jia},
title = {Characterizing pseudoentropy and simplifying pseudorandom generator constructions},
year = {2012},
isbn = {9781450312455},
publisher = {Association for Computing Machinery},
address = {New York, NY, USA},
url = {https://doi.org/10.1145/2213977.2214051},
doi = {10.1145/2213977.2214051},
booktitle = {Proceedings of the Forty-Fourth Annual ACM Symposium on Theory of Computing},
pages = {817–836},
numpages = {20},
location = {New York, New York, USA},
series = {STOC '12}
}

@InProceedings{haitner,
  author =	{Haitner, Iftach and Mazor, Noam and Silbak, Jad},
  title =	{{Incompressiblity and Next-Block Pseudoentropy}},
  booktitle =	{14th Innovations in Theoretical Computer Science Conference (ITCS 2023)},
  pages =	{66:1--66:18},
  series =	{Leibniz International Proceedings in Informatics (LIPIcs)},
  ISBN =	{978-3-95977-263-1},
  ISSN =	{1868-8969},
  year =	{2023},
  volume =	{251},
  editor =	{Tauman Kalai, Yael},
  publisher =	{Schloss Dagstuhl -- Leibniz-Zentrum f{\"u}r Informatik},
  address =	{Dagstuhl, Germany},
  URL =		{https://drops.dagstuhl.de/entities/document/10.4230/LIPIcs.ITCS.2023.66},
  URN =		{urn:nbn:de:0030-drops-175697},
  doi =		{10.4230/LIPIcs.ITCS.2023.66}
}

@inproceedings{il89,
  title={One-way functions are essential for complexity based cryptography},
  author={Impagliazzo, Russell and Luby, Michael},
  booktitle={30th Annual Symposium on Foundations of Computer Science},
  pages={230--235},
  year={1989},
  organization={IEEE Computer Society}
}

@article{hill,
author = {H{\r a}stad, Johan and Impagliazzo, Russell and Levin, Leonid A. and Luby, Michael},
title = {{A Pseudorandom Generator from any One-way Function}},
journal = {SIAM Journal on Computing},
volume = {28},
number = {4},
pages = {1364-1396},
year = {1999},
doi = {10.1137/S0097539793244708}
}

@INPROCEEDINGS{yao,
  author={Yao, Andrew C.},
  booktitle={23rd Annual Symposium on Foundations of Computer Science (FOCS 1982)}, 
  title={{Theory and Application of Trapdoor Functions}}, 
  year={1982},
  volume={},
  number={},
  pages={80-91},
  doi={10.1109/SFCS.1982.45}}

@INPROCEEDINGS{hirahara,
  author={Hirahara, Shuichi and Lu, Zhenjian and Nanashima, Mikito},
  booktitle={2024 IEEE 65th Annual Symposium on Foundations of Computer Science (FOCS)}, 
  title={{Optimal Coding for Randomized Kolmogorov Complexity and Its Applications}}, 
  year={2024},
  volume={},
  number={},
  pages={369-378},
  doi={10.1109/FOCS61266.2024.00030}}

\end{document}